\documentclass[sigplan,10pt]{acmart}
\acmConference[]{ACM SIGPLAN Conference on Programming Languages}{}{}
\acmYear{2018}
\acmISBN{} 
\acmDOI{} 
\startPage{1}

\setcopyright{none}

\usepackage{booktabs}   
\usepackage{subcaption} 

\usepackage{isabelle,isabellesym}
\usepackage{mathtools}

\newcommand{\extR}{\overline{\mathbb{R}}}
\renewcommand{\Re}{\operatorname{Re}}
\renewcommand{\Im}{\operatorname{Im}}

\newcommand{\Ind}{\operatorname{Ind}}
\newcommand{\Indp}{\operatorname{Indp}}

\newcommand{\jump}{\mathrm{jump}}

\newcommand{\SRemS}{\mathrm{SRemS}}
\newcommand{\SRemSE}{\mathrm{SRemSE}}

\newcommand{\Var}{\mathrm{Var}}
\newcommand{\Der}{\mathrm{Der}}
\newcommand{\Num}{\mathrm{Num}}
\newcommand{\NumD}{\mathrm{NumD}}

\renewcommand{\div}{\mathbin{\mathrm{div}}}
\renewcommand{\mod}{\mathbin{\mathrm{mod}}}

\newcommand{\lc}{\mathrm{lc}}

\begin{document}

\title[Counting Polynomial Roots]{Counting Polynomial Roots in Isabelle/HOL: A Formal Proof of the Budan-Fourier Theorem}


\author{Wenda Li}
\affiliation{
  \department{Department of Computer Science and Technology}              
  \institution{University of Cambridge}            
  \country{United Kingdom}                    
}
\email{wl302@cam.ac.uk}          

\author{Lawrence C. Paulson}
\affiliation{
	\department{Department of Computer Science and Technology}              
	\institution{University of Cambridge}            
	\country{United Kingdom}                    
}
\email{lp15@cam.ac.uk}         

\begin{abstract}
Many problems in computer algebra and numerical analysis can be reduced to counting or approximating the real roots of a polynomial within an interval. Existing verified root-counting procedures in major proof assistants are mainly based on the classical Sturm theorem, which only counts distinct roots. 

In this paper, we have strengthened the root-counting ability in Isabelle/HOL by first formally proving the Budan-Fourier theorem. Subsequently, based on Descartes' rule of signs and Taylor shift, we have provided a verified procedure to efficiently over-approximate the number of real roots within an interval, counting multiplicity. For counting multiple roots exactly, we have extended our previous formalisation of Sturm's theorem. Finally, we combine verified components in the developments above to improve our previous certified complex-root-counting procedures based on Cauchy indices. We believe those verified routines will be crucial for certifying programs and building tactics.
\end{abstract}



\maketitle

\section{Introduction}

 Counting the real and complex roots of a univariate polynomial has always been a fundamental task in computer algebra and numerical analysis. For example, given a routine that counts the number of real roots of a polynomial within an interval, we can compute each real root to arbitrary precision through bisection, in which sense we have solved the polynomial equation. Another example would be the Routh-Hurwitz stability criterion \cite[Section 23]{Arnold:1992tf}\cite[Chapter 9]{Marden:1949ui}, where the stability of a linear system can be tested by deciding the number of complex roots of the characteristic polynomial within the left half-plane (left of the imaginary axis).
 
 Numerous methods have been invented in the symbolic and numerical computing community to efficiently count (or test) real and complex roots of a polynomial \cite{Yap:2011in,Wilf:1978fy,Collins:1992kf}. However, in the theorem proving community, where procedures are usually formally verified in a foundational proof assistant (e.g., Coq, Isabelle, HOL and PVS), our choices are typically limited to relying on Sturm's theorem to count distinct roots within an interval through signed remainder sequences.
 
 In this paper, we aim to reinforce our root-counting ability in the Isabelle theorem prover \cite{paulson1994isabelle}. In particular, our main contributions are the following:
 \begin{itemize}
 	\item We have mechanised a proof of the Budan-Fourier theorem and a subsequent roots test based on Descartes' rule of signs and Taylor shift. This roots test efficiently over-approximates the number of real roots within an open interval, counting multiplicity.  
 	\item We have made a novel extension to our previous formalisation of Sturm's theorem to count real roots \emph{with} multiplicity.
 	\item Benefited from the developments above, we have extended our previous verified complex-root-counting procedures to more general cases: zeros on the border are allowed when counting roots in a half-plane; we can now additionally count roots within a ball.
 \end{itemize}
All results of this paper have been formalised in Isabelle/HOL without using extra axioms, and the source code is available from the following URL:
\begin{center}
	\url{https://bitbucket.org/liwenda1990/src-cpp-2019}
\end{center}
To reuse the results in this paper, consult our entries in the Archive of Formal Proofs \cite{Budan_Fourier-AFP, Count_Complex_Roots-AFP}, which we will keep updating.

This paper continues as follows: after introducing some frequently used notations (\S\ref{sec:notations}), we first present a formal proof of the Budan-Fourier theorem (\S\ref{sec:budan_fourier}), which eventually leads to the Descartes roots test. Next, we extend our previous formalisation of the classical Sturm theorem to count multiple real roots (\S\ref{sec:extending_sturm}). As an application, we apply those newly-formalised results to improve our previous complex-root-counting procedures (\S\ref{sec:counting_complex_roots}). After that, we discuss related work (\S\ref{sec:related_work}) and some experiments (\S\ref{sec:experiments}). Finally, we conclude the paper in \S\ref{sec:conclusion}.

\section{Notations} \label{sec:notations}
Below, we will present definitions and proofs in both natural language and the formal language of Isabelle. Some common notations are as follows:
\begin{itemize}
	\item we often use \isa{p} and \isa{q} in our proof scripts to denote polynomials. However, when presenting them in mathematical formulas, we will switch to capitalised letters---$P$ and $Q$.
	\item \isa{poly p a} in Isabelle means $P(a)$: the value of the univariate polynomial $P$ evaluated at the point $a$.
	\item $\mu(a,P)$ denotes the multiplicity/order of $a$ as a root of the polynomial $P$. In Isabelle, this becomes \isa{order a p}.
	\item $\extR$ denotes the extended real numbers: 
	\[
			\extR = \mathbb{R} \cup \{-\infty,+\infty \}.
	\]
	\item $\Num_{\mathbb{R}}(P;S)$ and $\Num_{\mathbb{C}}(P;S)$ denote the number of real and complex roots of~$P$ \emph{counting multiplicity} within the set $S$, while in Isabelle they both correspond to \isa{proots\_count p s}---they are distinguished by  the type of the polynomial \isa{p}.
	\item Similarly, $\NumD_{\mathbb{R}}(P;S)$ and $\NumD_{\mathbb{C}}(P;S)$ denote the number of \emph{distinct} real and complex roots of $P$ within $S$. In Isabelle, they correspond to
\begin{isabelle}
	card (proots\_within p s)
\end{isabelle}
differentiated by the type of the polynomial \isa{p}.
\end{itemize}

Many of the theorems presented in this paper will involve sign variations, so we give a definition here:
\begin{definition}[Sign variations] \label{def:sign_variations}
	Given a list of real numbers $[a_0,a_1,..., a_n]$, we use $\Var([a_0,a_1,...,a_n])$ to denote the number of \emph{sign variations} after dropping zeros. Additionally, we abuse the notation by letting
	\[
	\begin{aligned}
	\Var([P_0,P_1,...,P_n];a) &= \Var([P_0 (a),P_1 (a),...,P_n (a)]) \\
	\Var([P_0,P_1,...,P_n];a,b) &= \Var([P_0,P_1,...,P_n];a) \\ 
	& \qquad - \Var([P_0,P_1,...,P_n];b),
	\end{aligned}
	\]
	where $[P_0,P_1,...,P_n]$ is a sequence of univariate polynomials, and $\Var([P_0,P_1,...,P_n];a)$ is interpreted as the number of sign variations of $[P_0,P_1,...,P_n]$ evaluated at $a$. Finally, given $P(x) = a_0 + a_1 x + \cdots + a_n x^n$, we also use $\Var(P)$ to denote sign variations of the coefficient sequence of $P$:
	\[
	\Var(P) = \Var([a_0,a_1,..., a_n]).
	\]
\end{definition}

\begin{example}
	By Definition \ref{def:sign_variations}, we can have calculations like the following :
	\[
	\begin{split}
			\Var([1,-2,0,3]) &= \Var([1,-2,3]) = 2,\\
			\Var([x^2,x -2];0,1) &= \Var([x^2,x -2];0) - \Var([x^2, x -2];1) \\
				& = \Var([0,-2]) - \Var([1,-1]) \\
				&= 0 - 1 = -1,\\	
			\Var(1-x^2+2 x^3) & = \Var([1,0,-1,2]) = 2. \\
	\end{split}
	\]
\end{example}

\section{From the Budan-Fourier Theorem to the Descartes Roots Test} \label{sec:budan_fourier}
In this section, we first formalise the proof of the Budan-Fourier theorem. We then apply this to derive Descartes' rule of signs, which effectively over-approximates the number of positive real roots (counting multiplicity) of a real polynomial by calculating the sign variations of its coefficient sequence. We then use this to show the \emph{Descartes roots test}\footnote{ There does not seem to be a uniform name for this test \cite{akritas2008various,akritas2008new,Collins:2002ko}---here we follow the one used in Arno Eigenwillig's PhD thesis \cite{,Eigenwillig:2008tw} where he refers this test as "the Descartes test for roots".}: given a polynomial $P \in \mathbb{R}[x]$ of degree $n$ and a bounded interval $I = (a,b)$, we can apply Descartes' rule of signs to a base-transformed polynomial
\begin{equation} \label{eq:taylor_shift}
P_I (x) = (x+1) ^ n P \left (\frac{a x + b}{x+1} \right ),
\end{equation}
to over-approximate the number of real roots of $P$ over $(a,b)$. Note that the base transformation (\ref{eq:taylor_shift}) is commonly referred as \emph{Taylor shift} in the literature \cite{Kobel:2016im}.

Our formal proof of the Budan-Fourier theorem and Des\-cartes' rule of signs roughly follows the textbook by Basu et al. \cite{Basu:2006bo}, while that of the Descartes roots test is inspired by various sources \cite{Collins:1976gx,Eigenwillig:2008tw, Kobel:2016im}.

\subsection{The Budan-Fourier Theorem}

\begin{definition}[Fourier sequence]
	Let $P$ be a univariate polynomial of degree $n$. The \emph{Fourier sequence} of $P$ is generated through polynomial derivatives:
	\[
		\Der(P) = [P, P',..., P^{(n)}].
	\]
\end{definition}

\begin{theorem}[The Budan-Fourier theorem]
	\label{thm:budan_fourier}
	Let $P \in \mathbb{R}[x]$, $a, b$ be two extended real numbers (i.e., $a, b \in \extR$) such that $a<b$. Through Fourier sequences and sign variations, the Budan-Fourier theorem over-approximates $\Num_{\mathbb{R}}(P;(a,b])$ and the difference is an even number:
	\begin{itemize}
		\item $\Var(\Der(P);a, b) \geq \Num_{\mathbb{R}}(P;(a,b])$
		\item and $\Var(\Der(P);a, b) - \Num_{\mathbb{R}}(P;(a,b])$ is even.
	\end{itemize}
\end{theorem}

To prove Theorem \ref{thm:budan_fourier}, the key idea is to examine sign variations near a root of $P$: 
\[ \Var(\Der(P);c-\epsilon,c) \quad\text{and}\quad\Var(\Der(P);c,c+\epsilon), \] 
where $c$ is a possible root of $P$ and $\epsilon$ is a small real number.
	 
Regarding $\Var(\Der(P);c-\epsilon,c)$, the property we have derived in Isabelle/HOL is the following:
\begin{lemma} [\isa{budan{\isacharunderscore}fourier{\isacharunderscore}aux{\isacharunderscore}left\isacharprime}]
	\label{thm:budan_fourier_aux_left'}
	\begin{isabelle}
		\isanewline
		\ \ \isakeyword{fixes}\ c\ d\isactrlsub {\isadigit{1}}{\isacharcolon}{\isacharcolon}real\ \isakeyword{and}\ p{\isacharcolon}{\isacharcolon}{\isachardoublequoteopen}real\ poly{\isachardoublequoteclose}\isanewline
		\ \ \isakeyword{assumes}\ {\isachardoublequoteopen}d\isactrlsub {\isadigit{1}}\ {\isacharless}\ c{\isachardoublequoteclose}\ \isakeyword{and}\ {\isachardoublequoteopen}p\ {\isasymnoteq}\ {\isadigit{0}}{\isachardoublequoteclose}\isanewline
		\ \ \isakeyword{assumes}\ non{\isacharunderscore}zero{\isacharcolon}{\isachardoublequoteopen}{\isasymforall}x{\isachardot}\ d\isactrlsub {\isadigit{1}}\ {\isasymle}\ x\ {\isasymand}\ x\ {\isacharless}\ c\ {\isasymlongrightarrow}\ \isanewline
		\ \ \ \ \ \ \ \ \ \ \ \ \ \ {\isacharparenleft}{\isasymforall}q\ {\isasymin}\ set\ {\isacharparenleft}pders\ p{\isacharparenright}{\isachardot}\ poly\ q\ x\ {\isasymnoteq}\ {\isadigit{0}}{\isacharparenright}{\isachardoublequoteclose}\isanewline
		\ \ \isakeyword{shows}\ {\isachardoublequoteopen}changes{\isacharunderscore}itv{\isacharunderscore}der\ d\isactrlsub {\isadigit{1}}\ c\ p\ {\isasymge}\ order\ c\ p{\isachardoublequoteclose}\ \isanewline
		\ \ \ \ \ \ \ \ \ {\isachardoublequoteopen}even\ {\isacharparenleft}changes{\isacharunderscore}itv{\isacharunderscore}der\ d\isactrlsub {\isadigit{1}}\ c\ p\ {\isacharminus}\ order\ c\ p{\isacharparenright}{\isachardoublequoteclose}
	\end{isabelle}
\end{lemma}
\noindent where 
\begin{itemize}
	\item \isa{order c p} is $\mu(c,P)$: the order/multiplicity of $c$ as a root of the polynomial \isa{p}, as we described in \S\ref{sec:notations};
	\item \isa{changes\_itv\_der d\isactrlsub {\isadigit{1}} c p} stands for $\Var(\Der(P);d_1,c)$;
	\item the assumption \isa{non\_zero} asserts that $Q(x) \neq 0$ (i.e., \isa{poly q x \isasymnoteq\ 0}) for all $x \in [d_1, c)$ and $Q \in \Der(P)$. Here, \isa{pders p} is the Fourier sequence (i.e., $\Der(P)$) and \isa{set (pders p)} converts this sequence/list into a set of polynomials.
\end{itemize}
Essentially, $d_1$ in Lemma \ref{thm:budan_fourier_aux_left'} can be considered as $c-\epsilon$, since $d_1$ is asserted to be closer to $c$ from the left than any root of polynomials in $\Der(P)$. Therefore, Lemma \ref{thm:budan_fourier_aux_left'} claims that $\Var(\Der(P);c-\epsilon,c)$ always exceeds $\mu(c, P)$ by an even number. 

\begin{proof}[Proof of Lemma \ref{thm:budan_fourier_aux_left'}]
	By induction on the degree of $P$. 
	For the base case (i.e., the degree of $P$ is zero), the proof is trivial since both $\Var(\Der(P);d_1,c)$ and $\mu(c,P)$ are equal to 0.
	
	For the inductive case, through the induction hypothesis, we have
	\begin{multline} \label{eq:budan_fourier_aux_left_1}
			\Var(\Der(P');d_1,c) \geq \mu(c,P') \\ 
			\wedge \mathrm{even}(\Var(\Der(P');d_1,c) - \mu(c,P')).
	\end{multline}
	First, we consider the case when $P(c)=0$. In this case, we can derive 
	\begin{equation}  \label{eq:budan_fourier_aux_left_2}
		\mu(c,P) = \mu(c,P') + 1,
	\end{equation}
	\begin{equation}  \label{eq:budan_fourier_aux_left_3}
		\Var(\Der(P);d_1) = \Var(\Der(P');d_1) + 1,
	\end{equation}
	\begin{equation}  \label{eq:budan_fourier_aux_left_4}
		\Var(\Der(P);c) = \Var(\Der(P');c).
	\end{equation}
	Combining (\ref{eq:budan_fourier_aux_left_1}), (\ref{eq:budan_fourier_aux_left_2}), (\ref{eq:budan_fourier_aux_left_3}) and (\ref{eq:budan_fourier_aux_left_4}) yields
	\begin{multline} 
			\Var(\Der(P);d_1,c) \geq \mu(c,P)\\
			 \wedge \mathrm{even}(\Var(\Der(P);d_1,c) - \mu(c,P)),
			 \label{eq:budan_fourier_aux_left_5}
	\end{multline}	
	which concludes the proof.
	
	As for $P(c) \neq 0$, we can similarly have
	\begin{multline} \label{eq:budan_fourier_aux_left_6}
		\Var(\Der(P);c) 
		\\= 
		\begin{dcases}
		\Var(\Der(P');c), & \mbox{if } P'(c+\epsilon)>0 \\
		& \quad \leftrightarrow P(c)>0,\\
		\Var(\Der(P');c) + 1, &  \mbox{ otherwise,}\\
		\end{dcases}
	\end{multline}
	\begin{multline} \label{eq:budan_fourier_aux_left_7}
		\Var(\Der(P);d_1) 
		\\= 
		\begin{dcases}
		\Var(\Der(P');d_1), & \mbox{if } \mathrm{even}(\mu(c,P')) \\
		& \quad \leftrightarrow	 P'(x+\epsilon)>0 \\
		& \quad \leftrightarrow P(c)>0,\\
		\Var(\Der(P');d_1) + 1, &  \mbox{ otherwise.}\\
		\end{dcases}
	\end{multline}
	where $\leftrightarrow$ is the equivalence function in propositional logic. By putting together (\ref{eq:budan_fourier_aux_left_2}), (\ref{eq:budan_fourier_aux_left_6}) and (\ref{eq:budan_fourier_aux_left_7}), we can derive (\ref{eq:budan_fourier_aux_left_5}) through case analysis, and conclude the whole proof.
\end{proof}

Considering $\Var(\Der(P);c,c+\epsilon)$, we have an analogous proposition: 
\begin{lemma}[\isa{budan{\isacharunderscore}fourier{\isacharunderscore}aux{\isacharunderscore}right}]
	\label{thm:budan_fourier_aux_right}
	\begin{isabelle}
		\isanewline
		\ \ \isakeyword{fixes}\ c\ d\isactrlsub {\isadigit{2}}{\isacharcolon}{\isacharcolon}real\ \isakeyword{and}\ p{\isacharcolon}{\isacharcolon}{\isachardoublequoteopen}real\ poly{\isachardoublequoteclose}\isanewline
		\ \ \isakeyword{assumes}\ {\isachardoublequoteopen}c\ {\isacharless}\ d{\isadigit{2}}{\isachardoublequoteclose}\ \isakeyword{and}\ {\isachardoublequoteopen}p\ {\isasymnoteq}\ {\isadigit{0}}{\isachardoublequoteclose}\isanewline
		\ \ \isakeyword{assumes}\ {\isachardoublequoteopen}{\isasymforall}x{\isachardot}\ c\ {\isacharless}\ x\ {\isasymand}\ x\ {\isasymle}\ d\isactrlsub {\isadigit{2}}\ {\isasymlongrightarrow}\ \isanewline
		\ \ \ \ \ \ \ \ \ \ \ \ \ \ {\isacharparenleft}{\isasymforall}q\ {\isasymin}\ set\ {\isacharparenleft}pders\ p{\isacharparenright}{\isachardot}\ poly\ q\ x\ {\isasymnoteq}\ {\isadigit{0}}{\isacharparenright}{\isachardoublequoteclose}\ \ \isanewline
		\ \ \isakeyword{shows}\ {\isachardoublequoteopen}changes{\isacharunderscore}itv{\isacharunderscore}der\ c\ d\isactrlsub {\isadigit{2}}\ p\ {\isacharequal}\ {\isadigit{0}}{\isachardoublequoteclose}
	\end{isabelle}
\end{lemma}
\noindent which indicates that $\Var(\Der(P);c,c+\epsilon) = \Var(\Der(P);c,d_2) = 0$, since $d_2$ can be treated as $c+\epsilon$.

\begin{proof}[Proof of Lemma \ref{thm:budan_fourier_aux_right}]
	Similar to that of Lemma \ref{thm:budan_fourier_aux_left'}: by induction on the degree of $P$ and case analysis.
\end{proof}

With Lemma \ref{thm:budan_fourier_aux_right}, we can generalise Lemma \ref{thm:budan_fourier_aux_left'} a bit by allowing $P(d_1) = 0$ in the assumption:
\begin{lemma} [\isa{budan{\isacharunderscore}fourier{\isacharunderscore}aux{\isacharunderscore}left}]
	\label{thm:budan_fourier_aux_left}
	\begin{isabelle}
		\isanewline
		\ \ \isakeyword{fixes}\ c\ d\isactrlsub {\isadigit{1}}{\isacharcolon}{\isacharcolon}real\ \isakeyword{and}\ p{\isacharcolon}{\isacharcolon}{\isachardoublequoteopen}real\ poly{\isachardoublequoteclose}\isanewline
		\ \ \isakeyword{assumes}\ {\isachardoublequoteopen}d\isactrlsub {\isadigit{1}}\ {\isacharless}\ c{\isachardoublequoteclose}\ \isakeyword{and}\ {\isachardoublequoteopen}p\ {\isasymnoteq}\ {\isadigit{0}}{\isachardoublequoteclose}\isanewline
		\ \ \isakeyword{assumes}\ non{\isacharunderscore}zero{\isacharcolon}{\isachardoublequoteopen}{\isasymforall}x{\isachardot}\ d\isactrlsub {\isadigit{1}}\ {\isacharless}\ x\ {\isasymand}\ x\ {\isacharless}\ c\ {\isasymlongrightarrow}\ \isanewline
		\ \ \ \ \ \ \ \ \ \ \ \ \ \ {\isacharparenleft}{\isasymforall}q\ {\isasymin}\ set\ {\isacharparenleft}pders\ p{\isacharparenright}{\isachardot}\ poly\ q\ x\ {\isasymnoteq}\ {\isadigit{0}}{\isacharparenright}{\isachardoublequoteclose}\isanewline
		\ \ \isakeyword{shows}\ {\isachardoublequoteopen}changes{\isacharunderscore}itv{\isacharunderscore}der\ d\isactrlsub {\isadigit{1}}\ c\ p\ {\isasymge}\ order\ c\ p{\isachardoublequoteclose}\ \isanewline
		\ \ \ \ \ \ \ \ \ {\isachardoublequoteopen}even\ {\isacharparenleft}changes{\isacharunderscore}itv{\isacharunderscore}der\ d\isactrlsub {\isadigit{1}}\ c\ p\ {\isacharminus}\ order\ c\ p{\isacharparenright}{\isachardoublequoteclose}
	\end{isabelle}
\end{lemma}
\begin{proof}
	Let $d=(d_1+c)/2$. Lemma \ref{thm:budan_fourier_aux_right} and \ref{thm:budan_fourier_aux_left'} respectively yield
	\begin{equation} \label{eq:budan_fourier_aux_left_8}
		\Var(\Der(P);d_1,d) = 0,
	\end{equation}
	\begin{multline} \label{eq:budan_fourier_aux_left_9}
			\Var(\Der(P);d,c) \geq \mu(c,P) \\ \wedge \mathrm{even}(\Var(\Der(P);d,c) - \mu(c,P)).
	\end{multline}
	Moreover, by definition,
	\begin{multline} \label{eq:budan_fourier_aux_left_10}
		\Var(\Der(P);d_1,c) \\ = \Var(\Der(P);d_1,d) + \Var(\Der(P);d,c).
	\end{multline}
	From (\ref{eq:budan_fourier_aux_left_8}), (\ref{eq:budan_fourier_aux_left_9}) and (\ref{eq:budan_fourier_aux_left_10}), the conclusion follows.
\end{proof}

Finally, we come to our mechanised statement of the bounded interval case of Theorem \ref{thm:budan_fourier}:
\begin{theorem}[\isa{budan{\isacharunderscore}fourier{\isacharunderscore}interval}]
	\label{thm:budan_fourier_interval}
	\begin{isabelle}
		\isanewline
		\ \ \isakeyword{fixes}\ a\ b{\isacharcolon}{\isacharcolon}real\ \isakeyword{and}\ p{\isacharcolon}{\isacharcolon}{\isachardoublequoteopen}real\ poly{\isachardoublequoteclose}\isanewline
		\ \ \isakeyword{assumes}\ {\isachardoublequoteopen}a\ {\isacharless}\ b{\isachardoublequoteclose}\ \isakeyword{and}\ {\isachardoublequoteopen}p\ {\isasymnoteq}\ {\isadigit{0}}{\isachardoublequoteclose}\isanewline
		\ \ \isakeyword{shows}\ {\isachardoublequoteopen}changes{\isacharunderscore}itv{\isacharunderscore}der\ a\ b\ p\ \isanewline
		\ \ \ \ \ \ \ \ \ \ \ \ {\isasymge}\ proots{\isacharunderscore}count\ p\ {\isacharbraceleft}x{\isachardot}\ a\ {\isacharless}\ x\ {\isasymand}\ x\ {\isasymle}\ b{\isacharbraceright}\ {\isasymand}\isanewline
		\ \ \ \ \ \ \ \ \ even\ {\isacharparenleft}changes{\isacharunderscore}itv{\isacharunderscore}der\ a\ b\ p\ \isanewline
		\ \ \ \ \ \ \ \ \ \ \ \ \ \ \ \ {\isacharminus}\ proots{\isacharunderscore}count\ p\ {\isacharbraceleft}x{\isachardot}\ a\ {\isacharless}\ x\ {\isasymand}\ x\ {\isasymle}\ b{\isacharbraceright}{\isacharparenright}{\isachardoublequoteclose}
	\end{isabelle}
\end{theorem}
\noindent where \isa{proots{\isacharunderscore}count\ p\ {\isacharbraceleft}x{\isachardot}\ a\ {\isacharless}\ x\ {\isasymand}\ x\ {\isasymle}\ b{\isacharbraceright}{\isacharparenright}} denotes $\Num_{\mathbb{R}}(P;(a,b])$, which is the number of real roots of $P$ (counting multiplicity) within the interval $(a,b]$.

\begin{proof}[Proof of Theorem \ref{thm:budan_fourier_interval}]
	By induction on the number of roots of polynomials in $\Der(P)$ within the interval $(a,b)$.
	For the base case, we have 
	\begin{equation} \label{eq:budan_fourier_interval_1}
		Q(x) \neq 0, \qquad \mbox{for all } x \in (a,b) \mbox{ and } Q \in \Der(P),
	\end{equation}
	and then, by Lemma \ref{thm:budan_fourier_aux_left},
	\begin{multline} \label{eq:budan_fourier_interval_2}
		\Var(\Der(P);a,b) \geq \mu(b,P) \\ \wedge \mathrm{even}(\Var(\Der(P);a,b) - \mu(b,P)).
	\end{multline}
	In addition, (\ref{eq:budan_fourier_interval_1}) also leads to
	\begin{equation} \label{eq:budan_fourier_interval_3}
		\Num_{\mathbb{R}}(P;(a,b]) = \mu(b,P).
	\end{equation}
	We finish the base case by combining (\ref{eq:budan_fourier_interval_2}) and (\ref{eq:budan_fourier_interval_3}).
	
	Regrading the inductive case, let $b'$ be the largest root within the interval $(a,b)$ of the polynomials from $\Der(P)$:
	\begin{equation} \label{eq:budan_fourier_interval_4}
		b' = \max \{ x \in (a,b) \mid \exists Q \in \Der(P).\ Q(x) = 0 \}.
	\end{equation}
	With the induction hypothesis, we have
	\begin{multline} \label{eq:budan_fourier_interval_5}
	\Num_{\mathbb{R}}(P;(a,b']) \leq \Var(\Der(P);a,b') \\ \wedge \mathrm{even}(	\Num_{\mathbb{R}}(P;(a,b']) -  \Var(\Der(P);a,b')).
	\end{multline}
	Also, considering there is no root of $P$ within $(b',b)$ (otherwise it will be larger than $b'$, contradicting (\ref{eq:budan_fourier_interval_4})), we have
	\begin{equation} \label{eq:budan_fourier_interval_6}
		\Num_{\mathbb{R}}(P;(a,b]) = \Num_{\mathbb{R}}(P;(a,b']) + \mu(b,P). 
	\end{equation}
	Finally, Lemma \ref{thm:budan_fourier_aux_left} yields
	\begin{multline} \label{eq:budan_fourier_interval_7}
		\Var(\Der(P);b',b) \geq \mu(b,P) \\ \wedge \mathrm{even}(\Var(\Der(P);b',b) - \mu(b,P)).
	\end{multline}
Putting together (\ref{eq:budan_fourier_interval_5}), (\ref{eq:budan_fourier_interval_6}), and (\ref{eq:budan_fourier_interval_7}) finishes the proof.
\end{proof}

Note that Theorem \ref{thm:budan_fourier_interval} only corresponds to the bounded interval case of Theorem \ref{thm:budan_fourier}. In the formal development, we also have versions for $a=-\infty$, $b=+\infty$ or both.

An interesting corollary of the Budan-Fourier theorem is that when all roots are real, the over-approxiamation (i.e., $\Var(\Der(P);a,b)$) becomes exact:
\begin{corollary} [\isa{budan{\isacharunderscore}fourier{\isacharunderscore}real}]
	\label{thm:budan_fourier_real}
	\begin{isabelle}
		\isanewline
		\ \ \isakeyword{fixes}\ a\ b{\isacharcolon}{\isacharcolon}real\ \isakeyword{and}\ p{\isacharcolon}{\isacharcolon}{\isachardoublequoteopen}real\ poly{\isachardoublequoteclose}\isanewline
		\ \ \isakeyword{assumes}\ {\isachardoublequoteopen}all{\isacharunderscore}roots{\isacharunderscore}real\ p{\isachardoublequoteclose}\isanewline
		\ \ \isakeyword{shows}\ \isanewline
		\ \ \ \ {\isachardoublequoteopen}proots{\isacharunderscore}count\ p\ {\isacharbraceleft}x{\isachardot}\ x\ {\isasymle}\ a{\isacharbraceright}\ {\isacharequal}\ changes{\isacharunderscore}le{\isacharunderscore}der\ a\ p{\isachardoublequoteclose}\isanewline
		\ \ \ \ {\isachardoublequoteopen}a\ {\isacharless}\ b\ {\isasymlongrightarrow}\ proots{\isacharunderscore}count\ p\ {\isacharbraceleft}x{\isachardot}\ a\ {\isacharless}\ x\ {\isasymand}\ x\ {\isasymle}\ b{\isacharbraceright}\ \isanewline
		\ \ \ \ \ \ \ \ \ \ \ \ \ \ \ \ \ \ \ \ \ \ \ \ \ \ \ \ {\isacharequal}\ changes{\isacharunderscore}itv{\isacharunderscore}der\ a\ b\ p{\isachardoublequoteclose}\isanewline
		\ \ \ \ {\isachardoublequoteopen}proots{\isacharunderscore}count\ p\ {\isacharbraceleft}x{\isachardot}\ b\ {\isacharless}\ x{\isacharbraceright}\ {\isacharequal}\ changes{\isacharunderscore}gt{\isacharunderscore}der\ b\ p{\isachardoublequoteclose}
	\end{isabelle}
\end{corollary}
\noindent where 
\begin{itemize}
	\item \isa{all{\isacharunderscore}roots{\isacharunderscore}real\ p} is formally defined as every \emph{complex} root of $P$ having a zero imaginary part, 
	\item \isa{changes{\isacharunderscore}le{\isacharunderscore}der\ a\ p} encodes $\Var(\Der(P);-\infty,a)$,
	\item \isa{changes{\isacharunderscore}itv{\isacharunderscore}der\ a\ b\ p} encodes $\Var(\Der(P);a,b)$,
	\item \isa{changes{\isacharunderscore}gt{\isacharunderscore}der\ b\ p} encodes $\Var(\Der(P);b,+\infty)$.
\end{itemize}
\begin{proof}[Proof of Corollary \ref{thm:budan_fourier_real}]
	Let 
	\begin{equation*}  \label{eq:budan_fourier_real_1}
	\begin{split}
		t_1 &= \Var(\Der(P);-\infty,a) - \Num_{\mathbb{R}}(P;(-\infty,a]) \\
		t_2 &= \Var(\Der(P);a,b) - \Num_{\mathbb{R}}(P;(a,b]) \\
		t_3 &= \Var(\Der(P);a,b) - \Num_{\mathbb{R}}(P;(b,+\infty)) \\
	\end{split}
	\end{equation*}
	As a result of Theorem \ref{thm:budan_fourier}, we have 
	\begin{equation} \label{eq:budan_fourier_real_2}
		t_1 \geq 0 \wedge t_2 \geq 0 \wedge t_3 \geq 0.
	\end{equation}
	Additionally, by the definition of $\Var$ we derive 
	\begin{multline}  \label{eq:budan_fourier_real_3}
		\Var(\Der(P);-\infty,a) + \Var(\Der(P);a,b) \\ + \Var(\Der(P);a,+\infty) = \deg(P), 
	\end{multline}
	and the assumption (i.e., all roots are real) brings us
	\begin{multline}  \label{eq:budan_fourier_real_4}
		 \Num_{\mathbb{R}}(P;(-\infty,a]) +  \Num_{\mathbb{R}}(P;(a,b])  \\
		 + \Num_{\mathbb{R}}(P;(b,+\infty)) = \deg(P).
	\end{multline}
	Joining (\ref{eq:budan_fourier_real_3}) with (\ref{eq:budan_fourier_real_4}) yields 
	\begin{equation}  \label{eq:budan_fourier_real_5}
		t_1 + t_2 + t_3 = 0.
	\end{equation}
	Finally, putting  (\ref{eq:budan_fourier_real_2}) and (\ref{eq:budan_fourier_real_5}) together concludes the proof.
\end{proof}

\subsection{Descartes' Rule of Signs}
Given $a, b \in \extR$, $a<b$ and a polynomial $P \in \mathbb{R}[x]$, the Budan-Fourier theorem (Theorem \ref{thm:budan_fourier}) in the previous section grants us an effective way to over-approximate $\Num_{\mathbb{R}}(P;(a,b])$ (by an even number) through calculating $\Var(\Der(P);a,b)$.

Nevertheless, the approximation $\Var(\Der(P);a,b)$ still requires calculating a Fourier sequence ($\Der(P)$) and a series of polynomial evaluations. When $a=0$ and $b=+\infty$, the approximation can be refined to counting the number of sign variations of the coefficient sequence of $P$, which requires almost no calculation! Approximating $\Num_{\mathbb{R}}(P;(0,+\infty))$ using $\Var(P)$ (rather than $\Var(\Der(P);0,+\infty)$) is the celebrated Descartes' rule of signs:
\begin{theorem}[\isa{descartes{\isacharunderscore}sign}]
	\label{thm:descartes_sign}
	\begin{isabelle}
		\isanewline
		\ \ \isakeyword{fixes}\ p{\isacharcolon}{\isacharcolon}{\isachardoublequoteopen}real\ poly{\isachardoublequoteclose}\isanewline
		\ \ \isakeyword{assumes}\ {\isachardoublequoteopen}p\ {\isasymnoteq}\ {\isadigit{0}}{\isachardoublequoteclose}\isanewline
		\ \ \isakeyword{shows}\ {\isachardoublequoteopen}changes\ {\isacharparenleft}coeffs\ p{\isacharparenright}\ \isanewline
		\ \ \ \ \ \ \ \ \ \ \ \ {\isasymge}\ proots{\isacharunderscore}count\ p\ {\isacharbraceleft}x{\isachardot}\ {\isadigit{0}}\ {\isacharless}\ x{\isacharbraceright}\ {\isasymand}\ \isanewline
		\ \ \ \ \ \ \ \ \ \ even\ {\isacharparenleft}changes\ {\isacharparenleft}coeffs\ p{\isacharparenright}\isanewline
		\ \ \ \ \ \ \ \ \ \ \ \ \ \ \ \ \ {\isacharminus}\ proots{\isacharunderscore}count\ p\ {\isacharbraceleft}x{\isachardot}\ {\isadigit{0}}\ {\isacharless}\ x{\isacharbraceright}{\isacharparenright}{\isachardoublequoteclose}
	\end{isabelle}
\end{theorem}
\noindent where \isa{changes\ {\isacharparenleft}coeffs\ p{\isacharparenright}} encodes $\Var(P)$---sign variations of the coefficient sequence of $P$.
\begin{proof}
	Let $P=a_0 + a_1 x + a_2 x^2 + \cdots + a_{n-1} x^{n-1} + a_n x^n$. $\Der(P)$ is as follows:
	\begin{equation} \label{eq:descartes_sign_1}
		\begin{split}
		[ & a_0 + a_1 x + a_2 x^2 + \cdots + x_{n-1} x^{n-1}+ a_n x^n,\\
			& a_1 + 2 a_2 x + \cdots + (n-1) a_{n-1} x^{n-1}  + n a_n x^{n-1},\\
			& \vdots \\
			& (n-1)! a_{n-1} + n! a_n x\\
			& 	n! a_n \\
		]
		\end{split}
	\end{equation}
	where $n!$ is the factorial of $n$.
	From (\ref{eq:descartes_sign_1}), it can be derived that $\Der(P)$ has no sign variation when evaluated at $+\infty$:
	\begin{equation}  \label{eq:descartes_sign_2}
		\Var(\Der(P);+\infty) = 0.
	\end{equation}
	Also, evaluating $\Der(P)$ at $0$ gives $[a_0,a_1,...,(n-1)! a_{n-1}, n! a_n]$, hence its sign variations should equal $[a_0,a_1,...,a_{n-1}, a_n]$:
	\begin{equation}  \label{eq:descartes_sign_3}
		\Var(\Der(P);0) = \Var(P).
	\end{equation}
	Joining (\ref{eq:descartes_sign_2}) and (\ref{eq:descartes_sign_3}) gives $\Var(\Der(P);0,+\infty) = \Var(P)$, with which we apply Theorem \ref{thm:budan_fourier} to finish the proof.
\end{proof}

\subsection{Base Transformation and Taylor Shift} \label{secsub:base_transformation}
Given $P \in \mathbb{R}[x]$, with Descartes' rule of signs (Theorem \ref{thm:descartes_sign}) in the previous section, we can efficiently approximate $\Num_{\mathbb{R}}(P;(0,+\infty))$. However, in many cases, we are interested in roots within a bounded interval: $\Num_{\mathbb{R}}(P;I)$, where $I=(a,b)$ and $a,b \in \mathbb{R}$. Can we still exploit the efficiency from Descartes' rule of signs? The answer is yes, via an operation to transform $P$ into $P_I \in \mathbb{R}[x]$ such that
\begin{equation} \label{eq:taylor_shift_2}
	\Num_{\mathbb{R}}(P;I) = \Num_{\mathbb{R}}(P_I;(0,+\infty)).
\end{equation}
In order to develop this transformation operation, we first define the composition of a univariate polynomial and a rational function (i.e., a function of the form $f(x)=P(x)/Q(x)$, where $P$ and $Q$ are polynomials) in Isabelle/HOL:
\begin{isabelle}
	\isacommand{definition}\isamarkupfalse%
	\ fcompose{\isacharcolon}{\isacharcolon}\isanewline
	\ \ \ {\isachardoublequoteopen}{\isacharprime}a\ {\isacharcolon}{\isacharcolon}field\ poly\ {\isasymRightarrow}\ {\isacharprime}a\ poly\ {\isasymRightarrow}\ {\isacharprime}a\ poly\ {\isasymRightarrow}\ {\isacharprime}a\ poly{\isachardoublequoteclose}\ \isanewline
	\ \ \isakeyword{where}\isanewline
	\ \ \ \ {\isachardoublequoteopen}fcompose\ p\ q\isactrlsub {\isadigit{1}}\ q\isactrlsub {\isadigit{2}}\ {\isacharequal}\ \isanewline
	\ \ \ \ \ \ \ \ fst\ {\isacharparenleft}fold{\isacharunderscore}coeffs\ {\isacharparenleft}{\isasymlambda}a\ {\isacharparenleft}r\isactrlsub {\isadigit{1}}{\isacharcomma}r\isactrlsub {\isadigit{2}}{\isacharparenright}{\isachardot}\ \isanewline
	\ \ \ \ \ \ \ \ \ \ \ \ \ {\isacharparenleft}r\isactrlsub {\isadigit{2}}\ {\isacharasterisk}\ {\isacharbrackleft}{\isacharcolon}a{\isacharcolon}{\isacharbrackright}\ {\isacharplus}\ q\isactrlsub {\isadigit{1}}\ {\isacharasterisk}\ r\isactrlsub {\isadigit{1}}{\isacharcomma}q\isactrlsub {\isadigit{2}}\ {\isacharasterisk}\ r\isactrlsub {\isadigit{2}}{\isacharparenright}{\isacharparenright}\ p\ {\isacharparenleft}{\isadigit{0}}{\isacharcomma}{\isadigit{1}}{\isacharparenright}{\isacharparenright}{\isachardoublequoteclose}
\end{isabelle}
\noindent where
\begin{itemize}
	\item \isa{q\isactrlsub {\isadigit{1}}} and \isa{q\isactrlsub {\isadigit{2}}} are respectively the numerator and denominator of a rational function,
	\item \isa{fst} gives the first part of a pair,
	\item \isa{{\isacharbrackleft}{\isacharcolon}a{\isacharcolon}{\isacharbrackright}} is a constant polynomial lifted from the value $a$.
\end{itemize}
Also, \isa{fold{\isacharunderscore}coeffs} is the classical foldr operation on a coefficient sequence:
\begin{isabelle}
	\isacommand{definition}\isamarkupfalse%
	\ fold{\isacharunderscore}coeffs\ {\isacharcolon}{\isacharcolon}\ \isanewline
	\ \ \ \ {\isachardoublequoteopen}{\isacharparenleft}{\isacharprime}a{\isacharcolon}{\isacharcolon}zero\ {\isasymRightarrow}\ {\isacharprime}b\ {\isasymRightarrow}\ {\isacharprime}b{\isacharparenright}\ {\isasymRightarrow}\ {\isacharprime}a\ poly\ {\isasymRightarrow}\ {\isacharprime}b\ {\isasymRightarrow}\ {\isacharprime}b{\isachardoublequoteclose}\isanewline
	\ \ \isakeyword{where}\ {\isachardoublequoteopen}fold{\isacharunderscore}coeffs\ f\ p\ {\isacharequal}\ foldr\ f\ {\isacharparenleft}coeffs\ p{\isacharparenright}{\isachardoublequoteclose}
\end{isabelle}
Essentially, let $P$, $Q_1$, and $Q_2$ be three univariate polynomials over some field such that $P$ is of degree $n$. Our composition operation over these three polynomials (i.e., \isa{fcompose\ p\ q\isactrlsub {\isadigit{1}}\ q\isactrlsub {\isadigit{2}}}) gives the following polynomial:
\begin{equation} \label{eq:fcompose_1}
	(Q_2(x))^n P \left( \frac{Q_1(x)}{Q_2(x)} \right).
\end{equation}
The idea of (\ref{eq:fcompose_1}) can be illustrated by the following mechanised lemma:
\begin{lemma} [\isa{poly{\isacharunderscore}fcompose}]
	\label{thm:poly_fcompose}
	\begin{isabelle}
		\isanewline
		\ \ \isakeyword{fixes}\ p\ q\isactrlsub {\isadigit{1}}\ q\isactrlsub {\isadigit{2}}{\isacharcolon}{\isacharcolon}{\isachardoublequoteopen}{\isacharprime}a{\isacharcolon}{\isacharcolon}field\ poly{\isachardoublequoteclose}\isanewline
		\ \ \isakeyword{assumes}\ {\isachardoublequoteopen}poly\ q\isactrlsub {\isadigit{2}}\ x\ {\isasymnoteq}\ {\isadigit{0}}{\isachardoublequoteclose}\isanewline
		\ \ \isakeyword{shows}\ {\isachardoublequoteopen}poly\ {\isacharparenleft}fcompose\ p\ q\isactrlsub {\isadigit{1}}\ q\isactrlsub {\isadigit{2}}{\isacharparenright}\ x\ {\isacharequal}\ \isanewline
		\ \ \ \ \ \ \ \ \ \ \ \ \ \ \ \ \ \ \ \ \ \ poly\ p\ {\isacharparenleft}poly\ q\isactrlsub {\isadigit{1}}\ x\ {\isacharslash}\ poly\ q\isactrlsub {\isadigit{2}}\ x{\isacharparenright}\ \isanewline
		\ \ \ \ \ \ \ \ \ \ \ \ \ \ \ \ \ \ \ \ \ \ \ \ \ \ \ {\isacharasterisk}\ {\isacharparenleft}poly\ q\isactrlsub {\isadigit{2}}\ x{\isacharparenright}\ {\isacharcircum}\ {\isacharparenleft}degree\ p{\isacharparenright}{\isachardoublequoteclose}
	\end{isabelle}
\end{lemma}
\noindent where \isa{poly p x} gives the value of the polynomial \isa{p} when evaluated at \isa{x}.

When $Q_1(x) = a+b x$ and $Q_2(x) = 1+x$, (\ref{eq:fcompose_1}) yields a transformation (i.e., Taylor shift):
\begin{equation} \label{eq:taylor_shift_3}
	P_I (x) = (x+1) ^ n P \left (\frac{a x + b}{x+1} \right ),
\end{equation}
with which we have achieved (\ref{eq:taylor_shift_2}):
\begin{lemma}[\isa{proots{\isacharunderscore}count{\isacharunderscore}pos{\isacharunderscore}interval}]
	\label{thm:poly_fcompose_2}
	\begin{isabelle}
		\isanewline
		\ \ \isakeyword{fixes}\ a\ b{\isacharcolon}{\isacharcolon}real\ \isakeyword{and}\ p{\isacharcolon}{\isacharcolon}{\isachardoublequoteopen}real\ poly{\isachardoublequoteclose}\isanewline
		\ \ \isakeyword{assumes}\ {\isachardoublequoteopen}p\ {\isasymnoteq}\ {\isadigit{0}}{\isachardoublequoteclose}\ \isakeyword{and}\ {\isachardoublequoteopen}a\ {\isacharless}\ b{\isachardoublequoteclose}\isanewline
		\ \ \isakeyword{shows}\ {\isachardoublequoteopen}proots{\isacharunderscore}count\ p\ {\isacharbraceleft}x{\isachardot}\ a\ {\isacharless}\ x\ {\isasymand}\ x\ {\isacharless}\ b{\isacharbraceright}\ {\isacharequal}\ \isanewline
		\ \ \ \ \ \ \ \ \ \ \ \ proots{\isacharunderscore}count\ {\isacharparenleft}fcompose\ p\ {\isacharbrackleft}{\isacharcolon}b{\isacharcomma}a{\isacharcolon}{\isacharbrackright}\ {\isacharbrackleft}{\isacharcolon}{\isadigit{1}}{\isacharcomma}{\isadigit{1}}{\isacharcolon}{\isacharbrackright}{\isacharparenright}\isanewline
		\ \ \ \ \ \ \ \ \ \ \ \ \ \ \ \ \ \ \ \ \ \ \ \ \ \ \ \ \ \ \ \ \ \ \ \ \ \ \ \ \ \ \ {\isacharbraceleft}x{\isachardot}\ {\isadigit{0}}\ {\isacharless}\ x{\isacharbraceright}{\isachardoublequoteclose}
	\end{isabelle}
\end{lemma}
\noindent where 
\begin{itemize}
	\item \isa{{\isacharbrackleft}{\isacharcolon}b{\isacharcomma}a{\isacharcolon}{\isacharbrackright}} encodes the polynomial $b+a x$,
	\item \isa{{\isacharbrackleft}{\isacharcolon}{\isadigit{1}}{\isacharcomma}{\isadigit{1}}{\isacharcolon}{\isacharbrackright}} stands for the polynomial $1+x$.
\end{itemize}

\subsection{The Descartes Roots Test} \label{secsub:descartes_roots_test}
Finally, we come to the Descartes roots test. Given $P \in \mathbb{R}[x]$, $a, b \in \mathbb{R}$ and $I=(a,b)$, the Descartes roots test is $\Var(P_I)$: the number of sign variations on the coefficient sequence of the Taylor-shifted polynomial $P_I$:
\begin{isabelle}
	\isacommand{definition}\isamarkupfalse%
	\ descartes{\isacharunderscore}roots{\isacharunderscore}test{\isacharcolon}{\isacharcolon}\isanewline
	\ \ \ \ {\isachardoublequoteopen}real\ {\isasymRightarrow}\ real\ {\isasymRightarrow}\ real\ poly\ {\isasymRightarrow}\ nat{\isachardoublequoteclose}\ \isanewline
	\ \ \isakeyword{where}\isanewline
	\ \ \ \ {\isachardoublequoteopen}descartes{\isacharunderscore}roots{\isacharunderscore}test\ a\ b\ p\ {\isacharequal}\ nat\ {\isacharparenleft}changes\ \isanewline
	\ \ \ \ \ \ \ \ \ \ \ \ \ \ \ {\isacharparenleft}coeffs\ {\isacharparenleft}fcompose\ p\ {\isacharbrackleft}{\isacharcolon}b{\isacharcomma}a{\isacharcolon}{\isacharbrackright}\ {\isacharbrackleft}{\isacharcolon}{\isadigit{1}}{\isacharcomma}{\isadigit{1}}{\isacharcolon}{\isacharbrackright}{\isacharparenright}{\isacharparenright}{\isacharparenright}{\isachardoublequoteclose}
\end{isabelle}
where 
\begin{itemize}
	\item \isa{fcompose\ p\ {\isacharbrackleft}{\isacharcolon}b{\isacharcomma}a{\isacharcolon}{\isacharbrackright}\ {\isacharbrackleft}{\isacharcolon}{\isadigit{1}}{\isacharcomma}{\isadigit{1}}{\isacharcolon}{\isacharbrackright}} encodes Taylor shift as in (\ref{eq:taylor_shift_3}),
	\item \isa{coeffs} converts a polynomial into its coefficient sequence,
	\item \isa{changes} calculates the number of sign variations (i.e., $\Var$),
	\item \isa{nat} converts an integer into a natural number.
\end{itemize}

Just like $\Var(\Der(P);a,b)$, whose root approximation property has been reflected in Theorem \ref{thm:budan_fourier_interval}, $\Var(P_I)$ has a similar theorem related to $\Num_{\mathbb{R}}(P;(a,b))$:
\begin{theorem} [\isa{descartes{\isacharunderscore}roots{\isacharunderscore}test}]
	\label{thm:descartes_roots_test}
	\begin{isabelle}
		\isanewline
		\ \ \isakeyword{fixes}\ a\ b{\isacharcolon}{\isacharcolon}real\ \isakeyword{and}\ p{\isacharcolon}{\isacharcolon}{\isachardoublequoteopen}real\ poly{\isachardoublequoteclose}\isanewline
		\ \ \isakeyword{assumes}\ {\isachardoublequoteopen}p\ {\isasymnoteq}\ {\isadigit{0}}{\isachardoublequoteclose}\ \isakeyword{and}\ {\isachardoublequoteopen}a\ {\isacharless}\ b{\isachardoublequoteclose}\isanewline
		\ \ \isakeyword{shows}\ {\isachardoublequoteopen}proots{\isacharunderscore}count\ p\ {\isacharbraceleft}x{\isachardot}\ a\ {\isacharless}\ x\ {\isasymand}\ x\ {\isacharless}\ b{\isacharbraceright}\ \isanewline
		\ \ \ \ \ \ \ \ \ \ \ \ \ \ \ \ \ \ \ \ {\isasymle}\ descartes{\isacharunderscore}roots{\isacharunderscore}test\ a\ b\ p\ {\isasymand}\isanewline
		\ \ \ \ \ \ \ \ \ even\ {\isacharparenleft}descartes{\isacharunderscore}roots{\isacharunderscore}test\ a\ b\ p\ \isanewline
		\ \ \ \ \ \ \ \ \ \ \ \ \ \ \ \ \ \ {\isacharminus}\ proots{\isacharunderscore}count\ p\ {\isacharbraceleft}x{\isachardot}\ a\ {\isacharless}\ x\ {\isasymand}\ x\ {\isacharless}\ b{\isacharbraceright}{\isacharparenright}{\isachardoublequoteclose}
	\end{isabelle}
\end{theorem}
\noindent
which claims that $\Var(P_I)$ always exceeds $\Num_{\mathbb{R}}(P;(a,b))$ by an even number.

\begin{proof}[Proof of Theorem \ref{thm:descartes_roots_test}]
	Lemma \ref{thm:poly_fcompose_2} yields 
	\[
		\Num_{\mathbb{R}}(P;(a,b)) = \Num_{\mathbb{R}}(P_I;(0,+\infty)),
	\]
	with which we apply Theorem \ref{thm:descartes_sign} to conclude the proof.
\end{proof}

As an approximation, it is natural to ask when the Descartes roots test ($\Var(P_I)$)  is exact. From Theorem \ref{thm:descartes_roots_test}, it is easy to see that it would be exact as least when $\Var(P_I) = 0$ and $\Var(P_I) = 1$. Also, it is analogous to Corollary \ref{thm:budan_fourier_real} that $\Var(P_I)$ is exact when all roots are real:
\begin{corollary}[\isa{descartes{\isacharunderscore}roots{\isacharunderscore}test{\isacharunderscore}real}]
\label{thm:descartes_roots_test_real}
\begin{isabelle}
	\isanewline
	\ \ \isakeyword{fixes}\ a\ b{\isacharcolon}{\isacharcolon}real\ \isakeyword{and}\ p{\isacharcolon}{\isacharcolon}{\isachardoublequoteopen}real\ poly{\isachardoublequoteclose}\isanewline
	\ \ \isakeyword{assumes}\ {\isachardoublequoteopen}all{\isacharunderscore}roots{\isacharunderscore}real\ p{\isachardoublequoteclose}\ \isakeyword{and}\ {\isachardoublequoteopen}a\ {\isacharless}\ b{\isachardoublequoteclose}\ \isanewline
	\ \ \isakeyword{shows}\ {\isachardoublequoteopen}proots{\isacharunderscore}count\ p\ {\isacharbraceleft}x{\isachardot}\ a\ {\isacharless}\ x\ {\isasymand}\ x\ {\isacharless}\ b{\isacharbraceright}\ \isanewline
	\ \ \ \ \ \ \ \ \ \ \ \ \ \ \ \ \ \ \ \ \ \ \ \ \ \ {\isacharequal}\ descartes{\isacharunderscore}roots{\isacharunderscore}test\ a\ b\ p{\isachardoublequoteclose}
\end{isabelle}
\end{corollary}

\subsection{Remarks} \label{secsub:budan_fourier_remark}

Ever since the seminal paper by Collins and Akritas \cite{Collins:1976gx}, the Descartes roots test has been closely linked to modern real root isolation \cite{Kobel:2016im,Eigenwillig:2008tw,ROUILLIER200433}, where an effective method is needed for testing if an interval has zero or exactly one root. Although Sturm's theorem (which has already been formalised in Isabelle \cite{Sturm_Tarski-AFP,Li_CPP2016,Eberl:2015kb}) is also up to the task of root testing, it is considered too slow in modern computer algebra. Our mechanised version of the Descartes roots test is, by no means, state of the art; it is probably the most straightforward and naive implementation. Improvements over our current implementation are mainly about avoiding exact arithmetic, and the approaches include partial Taylor shift \cite{Kobel:2016im} and bitstream arithmetic \cite[Chapter 3]{Eigenwillig:2008tw}.

\section{Extending Sturm's Theorem to Exactly Count Multiple Roots} \label{sec:extending_sturm}

With the Descartes roots test we obtained from the previous section, we have an effective method to over-approximate the number of roots (with multiplicity) within an interval. However, we may sometimes want to know the exact number, as  we will describe below (\S\ref{sec:counting_complex_roots}). For now, we only have the classical Sturm theorem available (in Isabelle/HOL), which only counts \emph{distinct} real roots. In this section, we extend our previous formalisation of Sturm's theorem so that we will be able to count roots \emph{with multiplicity} and exactly. 

Our mechanised proof follows Rahman and Schmeisser \cite[Theorem 10.5.6]{Rahman:2016us}.

\begin{theorem}[Sturm's theorem] \label{thm:sturm}
Let $P \in \mathbb{R}[x]$, $a, b \in \extR$ such that $a<b$, $P(a) \neq 0$, and $P(b) \neq 0$. Sturm's theorem claims 
\[
	\NumD_{\mathbb{R}}(P;(a,b)) = \Var(\SRemS(P,P');a,b)
\]
where 
\begin{itemize}
	\item $\NumD_{\mathbb{R}}(P;(a,b))$ is the number of distinct roots of the polynomial $P$ within the interval $(a,b)$,
	\item $P'$ is the first derivative of $P$,
	\item $\Var$ is as in Definition \ref{def:sign_variations},
	\item $\SRemS(P,P')$ is the signed remainder sequence:
				\begin{equation} \label{eq:extended_sturm_1}
						[P_1,P_2,...,P_n],
				\end{equation}
				such that $P_1=P$, $P_2=P'$, $P_i = - (P_{i-1} \mod P_{i-2})$ ($3 \leq i \leq n$), and $P_n \mod P_{n-1} = 0$.
\end{itemize}
\end{theorem}

The core idea of our extended Sturm's theorem is to extend the remainder sequence ($\SRemS$):
\begin{definition}[Extended signed remainder sequence]
	Let $P, Q \in \mathbb{R}[x]$. The extended signed remainder sequence
	\begin{equation*}  \label{eq:extended_sturm_2}
			\SRemSE(P,Q) = [P_1,P_2,...,P_m]
	\end{equation*}
	is defined as $P_1=P$, $P_2=Q$, and for $i \geq 3$:
	\begin{equation} \label{eq:extended_sturm_3}
	P_i = 
	\begin{dcases}
	- (P_{i-1} \mod P_{i-2}), & \mbox{if } P_{i-1} \mod P_{i-2} \neq 0 \\
	P'_{i-1},&  \mbox{ otherwise,}\\
	\end{dcases}
	\end{equation}
	until $P_m$ such that $P_{m+1} = 0$ by (\ref{eq:extended_sturm_3}). 
\end{definition}

In Isabelle/HOL, $\SRemS$ and $\SRemSE$ are respectively mechanised as \isa{smods} and \isa{smods\_ext}:
\begin{isabelle}
	\isacommand{function}\isamarkupfalse%
	\ smods{\isacharcolon}{\isacharcolon}\ \isanewline
	\ \ \ \ {\isachardoublequoteopen}real\ poly\ {\isasymRightarrow}\ real\ poly\ {\isasymRightarrow}\ real\ poly\ list{\isachardoublequoteclose}\ \isanewline
	\ \ \isakeyword{where}\isanewline
	\ \ \ \ {\isachardoublequoteopen}smods\ p\ q\ {\isacharequal}\ {\isacharparenleft}if\ p\ {\isacharequal}\ {\isadigit{0}}\ then\ {\isacharbrackleft}{\isacharbrackright}\ \isanewline
	\ \ \ \ \ \ \ \ \ \ \ \ \ \ \ \ \ \ else\ p\ {\isacharhash}\ {\isacharparenleft}smods\ q\ {\isacharparenleft}{\isacharminus}\ {\isacharparenleft}p\ mod\ q{\isacharparenright}{\isacharparenright}{\isacharparenright}\isanewline
	\ \ \ \ \ \ \ \ \ \ \ \ \ \ \ \ \ \ {\isacharparenright}{\isachardoublequoteclose}
\end{isabelle}
\begin{isabelle}
	\isacommand{function}\isamarkupfalse%
	\ smods{\isacharunderscore}ext{\isacharcolon}{\isacharcolon}\isanewline
	\ \ \ \ {\isachardoublequoteopen}real\ poly\ {\isasymRightarrow}\ real\ poly\ {\isasymRightarrow}\ real\ poly\ list{\isachardoublequoteclose}\ \isanewline
	\ \ \isakeyword{where}\ \isanewline
	\ \ \ \ {\isachardoublequoteopen}smods{\isacharunderscore}ext\ p\ q\ {\isacharequal}\ \isanewline
	\ \ \ \ \ \ \ \ \ \ {\isacharparenleft}if\ p\ {\isacharequal}\ {\isadigit{0}}\ then\ \isanewline
	\ \ \ \ \ \ \ \ \ \ \ \ \ \ {\isacharbrackleft}{\isacharbrackright}\ \isanewline
	\ \ \ \ \ \ \ \ \ \ \ else\ if\ p\ mod\ q\ {\isasymnoteq}\ {\isadigit{0}}\ then\ \isanewline
	\ \ \ \ \ \ \ \ \ \ \ \ \ \ p\ {\isacharhash}\ {\isacharparenleft}smods{\isacharunderscore}ext\ q\ {\isacharparenleft}{\isacharminus}\ {\isacharparenleft}p\ mod\ q{\isacharparenright}{\isacharparenright}{\isacharparenright}\isanewline
	\ \ \ \ \ \ \ \ \ \ \ else\ \isanewline
	\ \ \ \ \ \ \ \ \ \ \ \ \ \ p\ {\isacharhash}\ {\isacharparenleft}smods{\isacharunderscore}ext\ q\ {\isacharparenleft}pderiv\ q{\isacharparenright}{\isacharparenright}\isanewline
	\ \ \ \ \ \ \ \ \ \ {\isacharparenright}{\isachardoublequoteclose}
\end{isabelle}
where \isa{[]} is an empty list and \isa{\isacharhash} is the Cons operation on lists---adding one element at the start of a list.

As $\SRemSE$ extends $\SRemS$ (from the back), it is natural to consider $\SRemS$ as a prefix of $\SRemSE$: 
\begin{lemma}[\isa{smods{\isacharunderscore}ext{\isacharunderscore}prefix}]
	\label{thm:smods_ext_prefix}
	\begin{isabelle}
		\isanewline
		\ \ \isakeyword{fixes}\ p\ q{\isacharcolon}{\isacharcolon}{\isachardoublequoteopen}real\ poly{\isachardoublequoteclose}\isanewline
		\ \ \isakeyword{defines}\ {\isachardoublequoteopen}r\ {\isasymequiv}\ last\ {\isacharparenleft}smods\ p\ q{\isacharparenright}{\isachardoublequoteclose}\ \isanewline
		\ \ \isakeyword{assumes}\ {\isachardoublequoteopen}p\ {\isasymnoteq}\ {\isadigit{0}}{\isachardoublequoteclose}\ \isakeyword{and}\ {\isachardoublequoteopen}q\ {\isasymnoteq}\ {\isadigit{0}}{\isachardoublequoteclose}\isanewline
		\ \ \isakeyword{shows}\ {\isachardoublequoteopen}smods{\isacharunderscore}ext\ p\ q\ {\isacharequal}\ smods\ p\ q\ \isanewline
		\ \ \ \ \ \ \ \ \ \ \ \ \ \ \ \ {\isacharat}\ tl\ {\isacharparenleft}smods{\isacharunderscore}ext\ r\ {\isacharparenleft}pderiv\ r{\isacharparenright}{\isacharparenright}{\isachardoublequoteclose}
	\end{isabelle}
\end{lemma}
\noindent where 
\begin{itemize}
	\item \isa{last} gives the last element of a list,
	\item \isa{@} concatenates two lists,
	\item \isa{tl} removes the head of a list,
	\item \isa{pderiv} returns the first derivative of a polynomial.
\end{itemize}

Moreover, we may need to realise that the last element of $\SRemS(P,Q)$ is actually the greatest common divisor (gcd) of $P$ and $Q$ up to some scalar: 
\begin{lemma}[\isa{last{\isacharunderscore}smods{\isacharunderscore}gcd}]
	\label{thm:last_smods_gcd}
	\begin{isabelle}
		\isanewline
		\ \ \isakeyword{fixes}\ p\ q{\isacharcolon}{\isacharcolon}{\isachardoublequoteopen}real\ poly{\isachardoublequoteclose}\isanewline
		\ \ \isakeyword{defines}\ {\isachardoublequoteopen}r\ {\isasymequiv}\ last\ {\isacharparenleft}smods\ p\ q{\isacharparenright}{\isachardoublequoteclose}\ \isanewline
		\ \ \isakeyword{assumes}\ {\isachardoublequoteopen}p\ {\isasymnoteq}\ {\isadigit{0}}{\isachardoublequoteclose}\isanewline
		\ \ \isakeyword{shows}\ {\isachardoublequoteopen}r\ {\isacharequal}\ smult\ {\isacharparenleft}lead{\isacharunderscore}coeff\ r{\isacharparenright}\ {\isacharparenleft}gcd\ p\ q{\isacharparenright}{\isachardoublequoteclose}
	\end{isabelle}
\end{lemma}
\noindent where 
\begin{itemize}
	\item \isa{smult} multiplies a polynomial with a scalar,
	\item \isa{lead\_coeff} gives the lead coefficient of a polynomial.
\end{itemize}

Finally, we can state (the bounded version of) our extended Sturm's theorem:
\begin{theorem} [\isa{sturm{\isacharunderscore}ext{\isacharunderscore}interval}]
	\label{thm:sturm_ext_interval}
	\begin{isabelle}
		\isanewline
		\ \ \isakeyword{fixes}\ a\ b{\isacharcolon}{\isacharcolon}real\ \isakeyword{and}\ p{\isacharcolon}{\isacharcolon}{\isachardoublequoteopen}real\ poly{\isachardoublequoteclose}\isanewline
		\ \ \isakeyword{assumes}\ {\isachardoublequoteopen}a\ {\isacharless}\ b{\isachardoublequoteclose}\ \isakeyword{and}\ {\isachardoublequoteopen}poly\ p\ a\ {\isasymnoteq}\ {\isadigit{0}}{\isachardoublequoteclose}\ \isanewline
		\ \ \ \ \ \ \isakeyword{and}\ {\isachardoublequoteopen}poly\ p\ b\ {\isasymnoteq}\ {\isadigit{0}}{\isachardoublequoteclose}\isanewline
		\ \ \isakeyword{shows}\ {\isachardoublequoteopen}proots{\isacharunderscore}count\ p\ {\isacharbraceleft}x{\isachardot}\ a\ {\isacharless}\ x\ {\isasymand}\ x\ {\isacharless}\ b{\isacharbraceright}\isanewline
		\ \ \ \ \ \ \ \ \ \ \ \ {\isacharequal}\ changes{\isacharunderscore}itv{\isacharunderscore}smods{\isacharunderscore}ext\ a\ b\ p\ {\isacharparenleft}pderiv\ p{\isacharparenright}{\isachardoublequoteclose}
	\end{isabelle}
\end{theorem}
\noindent where \isa{changes{\isacharunderscore}itv{\isacharunderscore}smods{\isacharunderscore}ext\ a\ b\ p\ {\isacharparenleft}pderiv\ p{\isacharparenright}} encodes $\Var(\SRemSE(P,P');a,b)$. Essentially, Theorem \ref{thm:sturm_ext_interval} claims that under some conditions
\begin{equation*} \label{eq:extended_sturm_4}
		\Num_{\mathbb{R}}(P;(a,b)) = \Var(\SRemSE(P,P');a,b).
\end{equation*}
\begin{proof}[Proof of Theorem \ref{thm:sturm_ext_interval}]
\sloppy
	By induction on the length of $\SRemSE(P,P')$, and case analysis on whether $P'=0$. 
	When $P'=0$, the proof is trivial since both $	\Num_{\mathbb{R}}(P;(a,b)) =0$ and $\Var(\SRemSE(P,P');a,b)=0$ provided $P \neq 0$.
	
	When $P' \neq 0$, we let $R$ be the last element of $\SRemS(P,P')$, and Lemma \ref{thm:last_smods_gcd} gives us
	\begin{equation} \label{eq:extended_sturm_41}
		R = \lc(R) \gcd(P,P'),
	\end{equation}
	where $\lc(R)$ is the leading coefficient of $R$.
	
	The essential part of the proof is to relate $\Num(P;(a,b))$ and $\Num(R;(a,b))$:
	\begin{align}
	 & \Num_{\mathbb{R}}(P;(a,b))  \notag\\
	 & = \sum_{x : P(x) = 0 \wedge x \in (a,b)} \mu(x,P) \label{eq:extended_sturm_5}\\
	&=  \sum_{x : P(x) = 0 \wedge x \in (a,b)} (1+\mu(x,R)) \label{eq:extended_sturm_6}\\ 
	&=  \NumD_{\mathbb{R}}(P;(a,b)) + \sum_{x : P(x) = 0 \wedge x \in (a,b)} \mu(x,R) \label{eq:extended_sturm_7}\\
	&=  \NumD_{\mathbb{R}}(P;(a,b)) + \sum_{x : R(x) = 0 \wedge x \in (a,b)} \mu(x,R) \label{eq:extended_sturm_8}\\
	&=  \NumD_{\mathbb{R}}(P;(a,b)) + \Num_{\mathbb{R}}(R;(a,b)). \label{eq:extended_sturm_9}
	\end{align}
	In particular, (\ref{eq:extended_sturm_6}) has been derived by
	\begin{multline*}
		\mu(x,P) = 1+\mu(x,P') = 1 + \min(\mu(x,P'),\mu(x,P)) \\ = 1 + \mu(x,\gcd(P,P')) = 1 + \mu(x,R),
	\end{multline*}
	 provided $P(x)=0$ and (\ref{eq:extended_sturm_41}). Also, (\ref{eq:extended_sturm_8}) is because $\{x \mid R(x)=0 \}  \subseteq  \{x \mid P(x)=0 \}$ and $\mu(y,R) = 0$ for all $y \in (a,b)$ such that $P(y) = 0$ and $R(y) \neq 0$. With (\ref{eq:extended_sturm_5}) - (\ref{eq:extended_sturm_9}), we have 
	\begin{equation} \label{eq:extended_sturm_10}
		\Num_{\mathbb{R}}(P;(a,b))  =  \NumD_{\mathbb{R}}(P;(a,b)) + \Num_{\mathbb{R}}(R;(a,b)).
	\end{equation}
	
	Moreover, the induction hypothesis yields 
	\begin{equation} \label{eq:extended_sturm_11}
		\Num_{\mathbb{R}}(R;(a,b))  = \Var(\SRemSE(R,R');a,b),
	\end{equation}
		and the classical Sturm theorem (Theorem \ref{thm:sturm}) yields
	\begin{equation} \label{eq:extended_sturm_12}
		\NumD_{\mathbb{R}}(P;(a,b))  = \Var(\SRemS(P,P');a,b).
	\end{equation}
	Also, by joining Lemma \ref{thm:smods_ext_prefix} and definition of $\Var$, we may have
	\begin{multline} \label{eq:extended_sturm_13}
		\Var(\SRemSE(P,P');a,b) = \Var(\SRemS(P,P');a,b) \\ + \Var(\SRemSE(R,R');a,b).
	\end{multline}
	
	Finally, putting together (\ref{eq:extended_sturm_10}), (\ref{eq:extended_sturm_11}), (\ref{eq:extended_sturm_12}), and (\ref{eq:extended_sturm_13}) yields
	\[
		\Num_{\mathbb{R}}(P;(a,b)) = \Var(\SRemSE(P,P');a,b),
	\]
	concluding the proof.
\end{proof}

Be aware that Lemma \ref{thm:sturm_ext_interval} only corresponds to the bounded version of the extended Sturm's theorem. Our formal development also contains unbounded versions (i.e., when $a=-\infty$ or $b=+\infty$).

\section{Applications to Counting Complex Roots} \label{sec:counting_complex_roots}

In the previous sections (\S\ref{sec:budan_fourier} and \S\ref{sec:extending_sturm}), we have demonstrated our enhancements for counting real roots in Isabelle/HOL. In this section, we will further apply those enhancements to improve existing complex-root-counting procedures  \cite{li_cauchy_index}. 

In particular, we will first review the idea of counting complex roots through Cauchy indices in \S\ref{secsub:cauchy_index}. After that, we will apply the extended Sturm's theorem (\S\ref{sec:extending_sturm}) to remove the constraint of forbidding roots on the border when counting complex roots in the upper half-plane (\S\ref{secsub:counting_upper}). In \S\ref{secsub:counting_ball}, we will combine the improved counting procedure (for roots in the upper half-plane) and  the base transformation in \S\ref{secsub:base_transformation} to build a verified procedure to count complex roots within a ball. Finally, we give some remarks about counting complex roots (\S\ref{secsub:complex_remark}).

\subsection{Number of Complex Roots and the Cauchy Index} \label{secsub:cauchy_index}

In this section we will briefly review the idea of counting complex roots through Cauchy indices, For a more detailed explanation, the reader can refer to our previous work \cite{li_cauchy_index}.

Thanks to the argument principle, the number of complex roots can be counted by evaluating a contour integral:
\begin{equation} \label{eq:cauchy_index_1}
\frac{1}{2 \pi i} \oint_\gamma \frac{P'(x)}{P(x) } d x = N,
\end{equation}
where $P \in \mathbb{C}[x]$, $P'(x)$ is the first derivative of $P$ and $N$ is the number of complex roots of $p$ (counting multiplicity) inside the loop $\gamma$. Also, by the definition of winding numbers, we have 
\begin{equation}  \label{eq:cauchy_index_2}
n(P \circ \gamma,0) = 	\frac{1}{2 \pi i} \oint_\gamma \frac{P'(x)}{P(x) } d x,
\end{equation}
where $\circ$ is function composition and $n(P \circ \gamma,0)$ is the winding number of the path $P \circ \gamma$ around $0$. Combining (\ref{eq:cauchy_index_1}) and  (\ref{eq:cauchy_index_2}) enables us to count (complex) roots by evaluating a winding number:
\begin{equation} \label{eq:cauchy_index_3}
N = n(P \circ \gamma,0).
\end{equation}
Now the question becomes how to evaluate the winding number $n(P \circ \gamma,0)$. One of the solutions is to utilise the Cauchy index.

To define the Cauchy index, we need to first introduce the concept of \emph{jumps}:
\begin{definition} [Jump] \label{def:jumpF}
	For $f : \mathbb{R} \rightarrow \mathbb{R}$ and $x \in \mathbb{R}$, we define
	\[
	\jump_+(f,x) =
	\begin{dcases}
	\frac{1}{2} & \mbox{if } \lim_{u \rightarrow x^+} f(u)=+\infty,\\
	-\frac{1}{2} & \mbox{if } \lim_{u \rightarrow x^+} f(u)=-\infty,\\
	0 &  \mbox{ otherwise.}\\
	\end{dcases}
	\]
	\[
	\jump_-(f,x) =
	\begin{dcases}
	\frac{1}{2} & \mbox{if } \lim_{u \rightarrow x^-} f(u)=+\infty,\\
	-\frac{1}{2} & \mbox{if } \lim_{u \rightarrow x^-} f(u)=-\infty,\\
	0 &  \mbox{ otherwise.}\\
	\end{dcases}
	\]
\end{definition}
We can now proceed to define Cauchy indices by summing up these jumps over an interval and along a path.
\begin{definition}[Cauchy index] \label{def:ind}
	For $f : \mathbb{R} \rightarrow \mathbb{R}$ and $a, b \in \extR$, the Cauchy index of $f$ over the interval $[a,b]$ is defined as
	\[
	\Ind_a^b(f) = \sum_{x \in [a,b)} \jump_+(f,x) \;-\! \sum_{x \in (a,b]} \jump_-(f,x).
	\]
\end{definition}
\begin{definition}[Cauchy index along a path] \label{def:indp}
	Given a path $\gamma : [0,1] \rightarrow \mathbb{C}$ and a point $z_0 \in \mathbb{C}$, the Cauchy index along $\gamma$ about $z_0$ is defined as
	\[
	\Indp(\gamma,z_0) = \Ind_0^1(f),
	\]
	where 
	\[
	f(t) = \frac{\Im(\gamma(t) - z_0)}{\Re(\gamma(t) - z_0).}
	\]
\end{definition}

As the Cauchy index $\Indp(\gamma,z_0)$ captures the way that $\gamma$ crosses the line $\{z \mid \Re(z)=\Re(z_0)\}$, we can evaluate the winding number through the Cauchy index:
\begin{theorem} \label{thm:winding_eq}
	Given a valid path $\gamma : [0,1] \rightarrow \mathbb{C}$ and a point $z_0 \in \mathbb{C}$, such that $\gamma$ is a loop and $z_0$ is not on the image of $\gamma$, we have 
	\[
	n(\gamma,z_0) = - \frac{\Indp(\gamma,z_0)}{2}.
	\]
\end{theorem}

Combining Theorem \ref{thm:winding_eq} and (\ref{eq:cauchy_index_3}) gives us a way to count complex polynomial roots:
\begin{equation} \label{eq:cauchy_index_4}
	N = - \frac{\Indp(P \circ \gamma,z_0)}{2}.
\end{equation}
What is more interesting is that $\Indp(P \circ \gamma,z_0)$ (or $\Ind^b_a(f, z_0)$) can be calculated through remainder sequences and sign variations when $P \circ \gamma$ (or $f$) is a rational function. That is, the right-hand side of (\ref{eq:cauchy_index_4}) becomes executable, and we have a procedure to count $N$.

\subsection{Resolving the Root-on-the-Border Issue when Counting Roots within a Half-Plane} \label{secsub:counting_upper}
Fundamentally, the complex-root-counting procedure in the previous section relies on the winding number and the argument principle, both of which disallow roots of $P$ on the border $\gamma$. As a result, both mechanised procedures --- counting roots within a rectangle and within a half-plane --- in our previous work will fail whenever there is a root on the border. In this section, we will utilise our newly mechanised extended Sturm's theorem to resolve the root-on-the-border issue when counting roots within a half-plane. Note that the root-on-the-border issue for the rectangular case, unfortunately, remains: we leave this issue for future work.

Considering that any half-plane can be transformed into the upper half-plane through a linear-transformation, we only need to focus on the upper-half-plane case:
\begin{isabelle}
	\isacommand{definition}\isamarkupfalse%
	\ proots{\isacharunderscore}upper\ {\isacharcolon}{\isacharcolon}\ {\isachardoublequoteopen}complex\ poly\ {\isasymRightarrow}\ nat{\isachardoublequoteclose}\ \isanewline
	\ \ \ \ \isakeyword{where}\isanewline
	\ \ {\isachardoublequoteopen}proots{\isacharunderscore}upper\ p\ {\isacharequal}\ proots{\isacharunderscore}count\ p\ {\isacharbraceleft}z{\isachardot}\ Im\ z\ {\isachargreater}\ {\isadigit{0}}{\isacharbraceright}{\isachardoublequoteclose}
\end{isabelle}
where \isa{proots{\isacharunderscore}upper\ p} encodes $\Num_{\mathbb{C}}(P;\{ z \mid \Im(z) > 0  \})$ --- the number of complex roots of $P$ within the upper half-plane $\{ z \mid \Im(z) > 0  \}$.

Previously, we relied on the following lemma to count $\Num_{\mathbb{C}}(P;\{ z \mid \Im(z) > 0  \})$:
\begin{lemma}[\isa{proots{\isacharunderscore}upper{\isacharunderscore}cindex{\isacharunderscore}eq}]
	\label{thm:proots_upper_cindex}
	\begin{isabelle}
		\isanewline
		\ \ \isakeyword{fixes}\ p{\isacharcolon}{\isacharcolon}{\isachardoublequoteopen}complex\ poly{\isachardoublequoteclose}\isanewline
		\ \ \isakeyword{assumes}\ {\isachardoublequoteopen}lead{\isacharunderscore}coeff\ p\ {\isacharequal}\ {\isadigit{1}}{\isachardoublequoteclose}\ \isanewline
		\ \ \ \ \ \ \isakeyword{and}\ no{\isacharunderscore}real{\isacharunderscore}roots{\isacharcolon}\ {\isachardoublequoteopen}{\isasymforall}x{\isasymin}proots\ p{\isachardot}\ Im\ x\ {\isasymnoteq}\ {\isadigit{0}}{\isachardoublequoteclose}\ \isanewline
		\ \ \ \ \isakeyword{shows}\ {\isachardoublequoteopen}proots{\isacharunderscore}upper\ p\ {\isacharequal}\ {\isacharparenleft}degree\ p\ {\isacharminus}\ \isanewline
		\ \ \ \ \ \ \ \ \ \ \ \ \ \ cindex{\isacharunderscore}poly{\isacharunderscore}ubd\ {\isacharparenleft}map{\isacharunderscore}poly\ Im\ p{\isacharparenright}\ \isanewline
		\ \ \ \ \ \ \ \ \ \ \ \ \ \ \ \ \ \ \ \ \ \ \ \ \ \ \ \ \ \ {\isacharparenleft}map{\isacharunderscore}poly\ Re\ p{\isacharparenright}{\isacharparenright}\ {\isacharslash}\ {\isadigit{2}}{\isachardoublequoteclose}
		\end{isabelle}
\end{lemma}
\noindent where 
\begin{itemize}
	\item \isa{lead{\isacharunderscore}coeff\ p\ {\isacharequal}\ {\isadigit{1}}} asserts the polynomial $P$ to be monic,
	\item the assumption \isa{no\_real\_roots} asserts that $P$ does not have any root on the real axis (i.e., the border). This assumption is, as mentioned earlier, because the argument principle disallows roots on the border.
	\item \isa{cindex{\isacharunderscore}poly{\isacharunderscore}ubd\ {\isacharparenleft}map{\isacharunderscore}poly\ Im\ p{\isacharparenright}\ {\isacharparenleft}map{\isacharunderscore}poly\ Re\ p{\isacharparenright}} encodes the Cauchy index \[\Ind^{+\infty}_{-\infty}\left( \lambda x. \frac{\Im(P(x))}{\Re(P(x))} \right ),
	\]
	which can be computed (through remainder sequences and sign variations) due to $\lambda x. \Im(P(x))/\Re(P(x))$ being a rational function.
\end{itemize}

To solve the root-on-the-border issue in Lemma \ref{thm:proots_upper_cindex}, we observe the effect of removing roots from a horizontal border:
\begin{lemma}[\isa{cindexE{\isacharunderscore}roots{\isacharunderscore}on{\isacharunderscore}horizontal{\isacharunderscore}border}]
	\label{thm:cindexE_roots_on_horizontal_border}
	\begin{isabelle}
		\isanewline
		\ \ \isakeyword{fixes}\ p\ q\ r\ {\isacharcolon}{\isacharcolon}{\isachardoublequoteopen}complex\ poly{\isachardoublequoteclose}\ \isakeyword{and}\ s\isactrlsub t{\isacharcolon}{\isacharcolon}complex\ \isanewline
		\ \ \ \ \isakeyword{and}\ a\ b\ s{\isacharcolon}{\isacharcolon}real\isanewline
		\ \ \isakeyword{defines}\ {\isachardoublequoteopen}{\isasymgamma}{\isasymequiv}linepath\ s\isactrlsub t\ {\isacharparenleft}s\isactrlsub t\ {\isacharplus}\ of{\isacharunderscore}real\ s{\isacharparenright}{\isachardoublequoteclose}\isanewline
		\ \ \isakeyword{assumes}\ {\isachardoublequoteopen}p\ {\isacharequal}\ q\ {\isacharasterisk}\ r{\isachardoublequoteclose}\ \isakeyword{and}\ {\isachardoublequoteopen}lead{\isacharunderscore}coeff\ r\ {\isacharequal}\ {\isadigit{1}}{\isachardoublequoteclose}\ \isanewline
		\ \ \ \ \ \ \isakeyword{and}\ {\isachardoublequoteopen}{\isasymforall}x{\isasymin}proots\ r{\isachardot}\ Im\ x\ {\isacharequal}\ Im\ s\isactrlsub t{\isachardoublequoteclose}\isanewline
		\ \ \ \ \isakeyword{shows}\ {\isachardoublequoteopen}cindexE\ a\ b\ {\isacharparenleft}{\isasymlambda}t{\isachardot}\ Im\ {\isacharparenleft}{\isacharparenleft}poly\ p\ {\isasymcirc}\ {\isasymgamma}{\isacharparenright}\ t{\isacharparenright}\ \isanewline
		\ \ \ \ \ \ \ \ \ \ \ \ \ \ \ \ \ \ \ \ \ \ \ \ \ \ {\isacharslash}\ Re\ {\isacharparenleft}{\isacharparenleft}poly\ p\ {\isasymcirc}\ {\isasymgamma}{\isacharparenright}\ t{\isacharparenright}{\isacharparenright}\ {\isacharequal}\isanewline
		\ \ \ \ \ \ \ \ \ \ cindexE\ a\ b\ {\isacharparenleft}{\isasymlambda}t{\isachardot}\ Im\ {\isacharparenleft}{\isacharparenleft}poly\ q\ {\isasymcirc}\ {\isasymgamma}{\isacharparenright}\ t{\isacharparenright}\ \isanewline
		\ \ \ \ \ \ \ \ \ \ \ \ \ \ \ \ \ \ \ \ \ \ \ \ \ \ {\isacharslash}\ Re\ {\isacharparenleft}{\isacharparenleft}poly\ q\ {\isasymcirc}\ {\isasymgamma}{\isacharparenright}\ t{\isacharparenright}{\isacharparenright}{\isachardoublequoteclose}
	\end{isabelle}
\end{lemma}
\noindent where the polynomial $Q$ is the result of $P$ after removing some roots on the horizontal border $\gamma$. Lemma \ref{thm:cindexE_roots_on_horizontal_border} claims that 
\begin{equation*}
	\Ind^a_b \left(\lambda x.  \frac{\Im(P(\gamma(x))}{\Re(P(\gamma(x))} \right) = \Ind^a_b \left(\lambda x.  \frac{\Im(Q(\gamma(x))}{\Re(Q(\gamma(x))} \right).
\end{equation*}
That is, the Cauchy index will remain the same if we only drop roots on a horizontal border.

We can now refine Lemma \ref{thm:proots_upper_cindex} by dropping the \isa{no\_real\_roots} assumption:
\begin{lemma}[\isa{proots{\isacharunderscore}upper{\isacharunderscore}cindex{\isacharunderscore}eq{\isacharprime}}]
	\label{thm:proots_upper_cindex'}
	\begin{isabelle}
		\isanewline
		\ \ \isakeyword{fixes}\ p{\isacharcolon}{\isacharcolon}{\isachardoublequoteopen}complex\ poly{\isachardoublequoteclose}\isanewline
		\ \ \isakeyword{assumes}\ {\isachardoublequoteopen}lead{\isacharunderscore}coeff\ p\ {\isacharequal}\ {\isadigit{1}}{\isachardoublequoteclose}\isanewline
		\ \ \isakeyword{shows}\ {\isachardoublequoteopen}proots{\isacharunderscore}upper\ p\ {\isacharequal}\ \isanewline
		\ \ \ \ \ \ \ \ \ \ \ \ \ \ {\isacharparenleft}degree\ p\ {\isacharminus}\ proots{\isacharunderscore}count\ p\ {\isacharbraceleft}x{\isachardot}\ Im\ x{\isacharequal}{\isadigit{0}}{\isacharbraceright}\isanewline
		\ \ \ \ \ \ \ \ \ \ \ \ \ \ {\isacharminus}\ cindex{\isacharunderscore}poly{\isacharunderscore}ubd\ {\isacharparenleft}map{\isacharunderscore}poly\ Im\ p{\isacharparenright}\ \isanewline
		\ \ \ \ \ \ \ \ \ \ \ \ \ \ \ \ \ \ \ \ \ \ \ \ \ \ \ \ \ \ \ \ {\isacharparenleft}map{\isacharunderscore}poly\ Re\ p{\isacharparenright}{\isacharparenright}\ {\isacharslash}{\isadigit{2}}{\isachardoublequoteclose}
	\end{isabelle}
\end{lemma}
To compare Lemma \ref{thm:proots_upper_cindex'} with Lemma \ref{thm:proots_upper_cindex}, we may note there is an extra term \isa{proots{\isacharunderscore}count\ p\ {\isacharbraceleft}x{\isachardot}\ Im\ x{\isacharequal}{\isadigit{0}}{\isacharbraceright}}  in the conclusion. This term encodes $\Num_{\mathbb{C}}(P;\{ z \mid \Im(z) = 0  \})$, and we can have 
\begin{multline} \label{eq:proots_upper_cindex_0}
	\Num_{\mathbb{C}}(P;\{ z \mid \Im(z) = 0  \}) \\= \Num_{\mathbb{R}}(\gcd(\Re(P),\Im(P));(-\infty,+\infty)),
\end{multline}
where $\Re(P), \Im(P) \in \mathbb{R}[x]$ are respectively the real and complex part of a complex polynomial $P$ such that $P(x) = \Re(P)(x) + i \Im(P)(x)$. The rationale behind (\ref{eq:proots_upper_cindex_0}) is that each root of $P$ on the real axis ($\{ z \mid \Im(z) = 0  \}$) is actually real and is also a root of both  $\Re(P)$ and  $\Im(P)$. More importantly, the right-hand side of (\ref{eq:proots_upper_cindex_0}) is where we will apply our extended Sturm's theorem in \S\ref{sec:extending_sturm}.

\begin{proof} [Proof of Lemma \ref{thm:proots_upper_cindex'}]
	Let $Q$ be the polynomial $P$ after removing all the roots on the border (i.e., the real axis) such that 
	\begin{equation} \label{eq:proots_upper_cindex_1}
		\Im (z) \neq 0 \quad \mbox{ whenever $z$ is a complex root of $Q$}. 
	\end{equation}
	By the definition of $Q$ and (\ref{eq:proots_upper_cindex_1}), we can apply Lemma \ref{thm:proots_upper_cindex} to derive
	\begin{multline} \label{eq:proots_upper_cindex_2}
		\Num_{\mathbb{C}}(P;\{ z \mid \Im(z) > 0  \}) \\ = \Num_{\mathbb{C}}(Q;\{ z \mid \Im(z) > 0  \})
			\\ =\deg(Q) - \Ind^{+\infty}_{-\infty } \left(\lambda x.  \frac{\Im(Q(x)}{\Re(Q(x)} \right).
	\end{multline}
	
	Moreover, $\deg(P)$ and $\deg(Q)$ are related by the fundamental theorem of algebra:
	\begin{equation} \label{eq:proots_upper_cindex_3}
		\deg(Q) = \deg(P) - \Num_{\mathbb{C}}(P;\{ z \mid \Im(z)  = 0  \}),
	\end{equation}
	and Lemma \ref{thm:cindexE_roots_on_horizontal_border} brings us the equivalence between two Cauchy indices:
	\begin{equation} \label{eq:proots_upper_cindex_4}
		\Ind^{+\infty}_{-\infty } \left(\lambda x.  \frac{\Im(Q(x)}{\Re(Q(x)} \right) = \Ind^{+\infty}_{-\infty } \left(\lambda x.  \frac{\Im(P(x)}{\Re(P(x)} \right).
	\end{equation}
	
	Putting (\ref{eq:proots_upper_cindex_2}), (\ref{eq:proots_upper_cindex_3}), and (\ref{eq:proots_upper_cindex_4}) together yields
	\begin{multline} \label{eq:proots_upper_cindex_5}
	\Num_{\mathbb{C}}(P;\{ z \mid \Im(z) > 0  \}) \\ = \deg(P) - \Num_{\mathbb{C}}(P;\{ z \mid \Im(z)  = 0  \}) \\ - \Ind^{+\infty}_{-\infty } \left(\lambda x.  \frac{\Im(P(x)}{\Re(P(x)} \right),
	\end{multline}
	which concludes the proof.
\end{proof}

Finally, we can have a \emph{code equation} (i.e., executable procedure) from the refined Lemma \ref{thm:proots_upper_cindex'} to compute $\Num_{\mathbb{C}}(P;\{ z \mid \Im(z) > 0  \})$:
\begin{lemma} [\isa{proots{\isacharunderscore}upper{\isacharunderscore}code{\isadigit{1}}{\isacharbrackleft}code{\isacharbrackright}}]
\label{thm:proots_upper_code}
\begin{isabelle}
	\isanewline
	\ \ {\isachardoublequoteopen}proots{\isacharunderscore}upper\ p\ {\isacharequal}\ \isanewline
	\ \ \ \ {\isacharparenleft}if\ p\ {\isasymnoteq}\ {\isadigit{0}}\ then\isanewline
	\ \ \ \ \ \ \ {\isacharparenleft}let\ p\isactrlsub m\ {\isacharequal}\ smult\ {\isacharparenleft}inverse\ {\isacharparenleft}lead{\isacharunderscore}coeff\ p{\isacharparenright}{\isacharparenright}\ p{\isacharsemicolon}\isanewline
	\ \ \ \ \ \ \ \ \ \ \ \ p\isactrlsub I\ {\isacharequal}\ map{\isacharunderscore}poly\ Im\ p\isactrlsub m{\isacharsemicolon}\isanewline
	\ \ \ \ \ \ \ \ \ \ \ \ p\isactrlsub R\ {\isacharequal}\ map{\isacharunderscore}poly\ Re\ p\isactrlsub m{\isacharsemicolon}\isanewline
	\ \ \ \ \ \ \ \ \ \ \ \ g\ {\isacharequal}\ gcd\ p\isactrlsub I\ p\isactrlsub R\isanewline
	\ \ \ \ \ \ \ \ in\isanewline
	\ \ \ \ \ \ \ \ \ \ \ \ nat\ {\isacharparenleft}{\isacharparenleft}degree\ p\ \isanewline
	\ \ \ \ \ \ \ \ \ \ \ \ \ \ \ \ \ \ \ \ {\isacharminus}\ changes{\isacharunderscore}R{\isacharunderscore}smods{\isacharunderscore}ext\ g\ {\isacharparenleft}pderiv\ g{\isacharparenright}\isanewline
	\ \ \ \ \ \ \ \ \ \ \ \ \ \ \ \ \ \ \ \ {\isacharminus}\ changes{\isacharunderscore}R{\isacharunderscore}smods\ p\isactrlsub R\ p\isactrlsub I{\isacharparenright}\ div\ {\isadigit{2}}\isanewline
	\ \ \ \ \ \ \ \ \ \ \ \ \ \ \ \ {\isacharparenright}\ \isanewline
	\ \ \ \ \ \ \ \ {\isacharparenright}\isanewline
	\ \ \ \ else\ \isanewline
	\ \ \ \ \ \ Code{\isachardot}abort\ {\isacharparenleft}STR\ {\isacharprime}{\isacharprime}proots{\isacharunderscore}upper\ fails\ when\ p{\isacharequal}{\isadigit{0}}{\isachardot}{\isacharprime}{\isacharprime}{\isacharparenright}\isanewline
	\ \ \ \ \ \ \ \ \ \ \ \ \ \ \ \ \ \ \ \ \ \ \ \ \ \ \ \ \ \ \ \ \ {\isacharparenleft}{\isasymlambda}{\isacharunderscore}{\isachardot}\ proots{\isacharunderscore}upper\ p{\isacharparenright}{\isacharparenright}{\isachardoublequoteclose}
\end{isabelle}
\end{lemma}
\noindent where 
\begin{itemize}
	\item \isa{p\isactrlsub m} is a monic polynomial produced by $P$ divided by its leading coefficient, and this monic polynomial is required by the assumption \isa{lead{\isacharunderscore}coeff\ p\ {\isacharequal}\ {\isadigit{1}}} in Lemma \ref{thm:proots_upper_cindex'}.
	\item \isa{changes{\isacharunderscore}R{\isacharunderscore}smods{\isacharunderscore}ext\ g\ {\isacharparenleft}pderiv\ g{\isacharparenright}} encodes 
	\begin{multline*}
		\Var(\SRemSE(G,G');-\infty,+\infty) \\ \mbox{where $G=\gcd(\Re(P),\Im(P))$},
	\end{multline*}
	which computes $\Num_{\mathbb{C}}(P;\{ z \mid \Im(z)  = 0  \}) $ in Lemma \ref{thm:proots_upper_cindex'}. Note that this part is where our extended Sturm's theorem in \S\ref{sec:extending_sturm} has been utilised.
	\item \isa{changes{\isacharunderscore}R{\isacharunderscore}smods\ p\isactrlsub R\ p\isactrlsub I} stands for 
\[ \Var(\SRemS(P_R,P_I);-\infty,+\infty),\] 
which computes $\Ind^{+\infty}_{-\infty } \left(\lambda x.  \frac{\Im(P(x)}{\Re(P(x)} \right)$ in Lemma \ref{thm:proots_upper_cindex'}.
	\item The command \isa{Code{\isachardot}abort} raises an exception in the case of $P = 0$.
\end{itemize}
Overall, Lemma \ref{thm:proots_upper_code} asserts that 
\[ \Num_{\mathbb{C}}(P;\{ z \mid \Im(z)  > 0  \})\] is equivalent to an executable expression: the right-hand side of Lemma \ref{thm:proots_upper_code}. Because it is declared as a code equation, it implements a verified procedure to compute $\Num_{\mathbb{C}}(P;\{ z \mid \Im(z)  > 0  \})$. We can now type the following command in Isabelle/HOL:
\begin{isabelle}
	\isacommand{value}\isamarkupfalse%
	\ {\isachardoublequoteopen}proots{\isacharunderscore}upper\ {\isacharbrackleft}{\isacharcolon}{\isadigit{1}}{\isacharplus}{\isasymi}{\isacharcomma}\ {\isacharminus}{\isadigit{2}}{\isacharminus}{\isasymi}{\isacharcomma}\ {\isadigit{1}}{\isacharcolon}{\isacharbrackright}{\isachardoublequoteclose}
\end{isabelle}
to compute $\Num_{\mathbb{C}}((1+i) +( -2 - i) x + x^2;\{ z \mid \Im(z)  > 0 \})$, which was not possible in our previous work \cite{li_cauchy_index} since the polynomial $(1+i) +( -2 - i) x + x^2 = (x-1) (x-1-i)$ has a root on the border (i.e., the real axis).

\subsection{Counting Roots within a Ball} \label{secsub:counting_ball}

In this section, we will introduce a verified procedure to count $\Num_{\mathbb{C}}(P;\{ z \mid \left | z - z_0 \right |   < r \})$, the number of complex roots of a polynomial $P$ within the ball $\{ z \mid \left | z - z_0 \right |   < r \}$. The core idea is to use the base transformation in \S\ref{secsub:base_transformation} to convert the current case to the one of counting roots within the upper half-plane, and then make use of the procedure in \S\ref{secsub:counting_upper} to finish counting.

Let \isa{proots{\isacharunderscore}ball} denote $\Num_{\mathbb{C}}(P;\{ z \mid \left | z - z_0 \right |  < r \})$:
\begin{isabelle}
	\isacommand{definition}\isamarkupfalse%
	\ proots{\isacharunderscore}ball{\isacharcolon}{\isacharcolon}\isanewline
	\ \ \ \ {\isachardoublequoteopen}complex\ poly\ {\isasymRightarrow}\ complex\ {\isasymRightarrow}\ real\ {\isasymRightarrow}\ nat{\isachardoublequoteclose}\ \isanewline
	\ \ \isakeyword{where}\ \isanewline
	\ \ \ \ {\isachardoublequoteopen}proots{\isacharunderscore}ball\ p\ z\isactrlsub {\isadigit{0}}\ r\ {\isacharequal}\ proots{\isacharunderscore}count\ p\ {\isacharparenleft}ball\ z\isactrlsub {\isadigit{0}}\ r{\isacharparenright}{\isachardoublequoteclose}
\end{isabelle}

With the transformation operation (\isa{fcompose}) we developed in \S\ref{secsub:base_transformation}, we can derive the following equivalence relation in the number of roots:
\begin{lemma}[\isa{proots{\isacharunderscore}ball{\isacharunderscore}plane{\isacharunderscore}eq}]
	\label{thm:proots_ball_plane_eq}
	\begin{isabelle}
		\isanewline
		\ \ \isakeyword{fixes}\ p{\isacharcolon}{\isacharcolon}{\isachardoublequoteopen}complex\ poly{\isachardoublequoteclose}\isanewline
		\ \ \isakeyword{assumes}\ {\isachardoublequoteopen}p\ {\isasymnoteq}\ {\isadigit{0}}{\isachardoublequoteclose}\isanewline
		\ \ \isakeyword{shows}\ {\isachardoublequoteopen}proots{\isacharunderscore}count\ p\ {\isacharparenleft}ball\ {\isadigit{0}}\ {\isadigit{1}}{\isacharparenright}\ \isanewline
		\ \ \ \ \ \ \ \ \ \ \ \ {\isacharequal}\ proots{\isacharunderscore}count\ {\isacharparenleft}fcompose\ p\ {\isacharbrackleft}{\isacharcolon}{\isasymi}{\isacharcomma}{\isacharminus}{\isadigit{1}}{\isacharcolon}{\isacharbrackright}\ {\isacharbrackleft}{\isacharcolon}{\isasymi}{\isacharcomma}{\isadigit{1}}{\isacharcolon}{\isacharbrackright}{\isacharparenright}\isanewline
		\ \ \ \ \ \ \ \ \ \ \ \ \ \ \ \ \ \ \ \ \ \ \ \ \ \ \ \ \ \ \ \ \ \ \ \ \ \ \ \ \ \ {\isacharbraceleft}z{\isachardot}\ {\isadigit{0}}\ {\isacharless}\ Im\ z{\isacharbraceright}{\isachardoublequoteclose}
	\end{isabelle}	
\end{lemma}
\noindent That is,
\begin{multline} \label{eq:proots_count_ball_1}
	\Num_{\mathbb{C}}(P; \{ z \mid \left | z  \right |  < 1 \}) \\ = 	\Num_{\mathbb{C}} \left( (i+x)^n P\left( \frac{i-x}{i+x} \right); \{ z \mid \Im(z) > 0 \} \right),
\end{multline}
where $n$ is the degree of $P$.

Moreover, we can relate roots between different balls using normal polynomial composition:
\begin{lemma}[\isa{proots{\isacharunderscore}uball{\isacharunderscore}eq}]
	\label{thm:proots_uball_eq}
	\begin{isabelle}
		\isanewline
		\ \ \isakeyword{fixes}\ p{\isacharcolon}{\isacharcolon}{\isachardoublequoteopen}complex\ poly{\isachardoublequoteclose}\ \isakeyword{and}\ z\isactrlsub {\isadigit{0}}{\isacharcolon}{\isacharcolon}complex\ \isakeyword{and}\ r{\isacharcolon}{\isacharcolon}real\isanewline
		\ \ \isakeyword{assumes}\ {\isachardoublequoteopen}p\ {\isasymnoteq}\ {\isadigit{0}}{\isachardoublequoteclose}\ \isakeyword{and}\ {\isachardoublequoteopen}r\ {\isachargreater}\ {\isadigit{0}}{\isachardoublequoteclose}\isanewline
		\ \ \isakeyword{shows}\ {\isachardoublequoteopen}proots{\isacharunderscore}count\ p\ {\isacharparenleft}ball\ z\isactrlsub {\isadigit{0}}\ r{\isacharparenright}\isanewline
		\ \ \ \ \ \ \ \ \ \ \ \ \ \ \ \ {\isacharequal}\ proots{\isacharunderscore}count\ {\isacharparenleft}p\ {\isasymcirc}\isactrlsub p\ {\isacharbrackleft}{\isacharcolon}z\isactrlsub {\isadigit{0}}{\isacharcomma}\ of{\isacharunderscore}real\ r{\isacharcolon}{\isacharbrackright}{\isacharparenright}\isanewline
		\ \ \ \ \ \ \ \ \ \ \ \ \ \ \ \ \ \ \ \ \ \ \ \ \ \ \ \ \ \ \ \ \ \ \ \ \ \ \ \ \ \ \ \ {\isacharparenleft}ball\ {\isadigit{0}}\ {\isadigit{1}}{\isacharparenright}{\isachardoublequoteclose}
	\end{isabelle}
\end{lemma}
\noindent where \isa{{\isasymcirc}\isactrlsub p} encodes the composition operation between two polynomials. Overall, Lemma \ref{thm:proots_uball_eq} claims
\begin{multline} \label{eq:proots_count_ball_2}
\Num_{\mathbb{C}}(P; \{ z \mid \left | z - z_0 \right |  < r \}) \\ = 	\Num_{\mathbb{C}} \left( P(r x - z_0); \{ z \mid \left | z  \right |  < 1 \} \right).
\end{multline}

Finally, we can derive a code equation for \isa{proots{\isacharunderscore}ball}:
\begin{lemma}[\isa{proots{\isacharunderscore}ball{\isacharunderscore}code{\isadigit{1}}{\isacharbrackleft}code{\isacharbrackright}}]
	\label{thm:proots_ball_code}
	\begin{isabelle}
		\isanewline
		\ \ {\isachardoublequoteopen}proots{\isacharunderscore}ball\ p\ z\isactrlsub {\isadigit{0}}\ r\ {\isacharequal}\ \isanewline
		\ \ \ \ \ \ {\isacharparenleft}\ if\ r\ {\isasymle}\ {\isadigit{0}}\ then\ \isanewline
		\ \ \ \ \ \ \ \ \ \ {\isadigit{0}}\isanewline
		\ \ \ \ \ \ \ \ else\ if\ p\ {\isasymnoteq}\ {\isadigit{0}}\ then\isanewline
		\ \ \ \ \ \ \ \ \ \ proots{\isacharunderscore}upper\ {\isacharparenleft}fcompose\ \isanewline
		\ \ \ \ \ \ \ \ \ \ \ \ \ {\isacharparenleft}p\ {\isasymcirc}\isactrlsub p\ {\isacharbrackleft}{\isacharcolon}z\isactrlsub {\isadigit{0}}{\isacharcomma}\ of{\isacharunderscore}real\ r{\isacharcolon}{\isacharbrackright}{\isacharparenright}\ {\isacharbrackleft}{\isacharcolon}{\isasymi}{\isacharcomma}{\isacharminus}{\isadigit{1}}{\isacharcolon}{\isacharbrackright}\ {\isacharbrackleft}{\isacharcolon}{\isasymi}{\isacharcomma}{\isadigit{1}}{\isacharcolon}{\isacharbrackright}{\isacharparenright}\isanewline
		\ \ \ \ \ \ \ \ else\ \isanewline
		\ \ \ \ \ \ \ \ \ \ Code{\isachardot}abort\ {\isacharparenleft}STR\ {\isacharprime}{\isacharprime}proots{\isacharunderscore}ball\ fails\ \isanewline
		\ \ \ \ \ \ \ \ \ \ \ \ \ \ \ \ \ when\ p{\isacharequal}{\isadigit{0}}{\isachardot}{\isacharprime}{\isacharprime}{\isacharparenright}\ {\isacharparenleft}{\isasymlambda}{\isacharunderscore}{\isachardot}\ proots{\isacharunderscore}ball\ p\ z\isactrlsub {\isadigit{0}}\ r{\isacharparenright}\isanewline
		\ \ \ \ \ \ {\isacharparenright}{\isachardoublequoteclose}
	\end{isabelle}
\end{lemma}
\noindent The idea behind of Lemma \ref{thm:proots_ball_code} is to combine (\ref{eq:proots_count_ball_1}) and (\ref{eq:proots_count_ball_2}):
\begin{multline}
	\Num_{\mathbb{C}}(P; \{ z \mid \left | z - z_0 \right |  < r \}) \\ = \Num_{\mathbb{C}} \left( (r x + i - z_0)^n P\left( \frac{-r x + i + z_0}{r x + i - z_0} \right); \{ z \mid \Im(z) > 0 \} \right),
\end{multline}
so that we can apply Lemma \ref{thm:proots_upper_code} to count roots within the upper half-plane instead.

Because Lemma \ref{thm:proots_ball_code} is declared as a code equation, we can execute it. For example, we can now type the following command in Isabelle/HOL:
\begin{isabelle}
	\isacommand{value}\isamarkupfalse%
	\ {\isachardoublequoteopen}proots{\isacharunderscore}ball\ {\isacharbrackleft}{\isacharcolon}{\isasymi}{\isacharcomma}{\isacharminus}\ {\isadigit{1}}\ {\isacharminus}\ {\isasymi}{\isacharcomma}\ {\isadigit{1}}{\isacharcolon}{\isacharbrackright}\ {\isadigit{0}}\ {\isadigit{1}}{\isachardoublequoteclose}
\end{isabelle}
to check that the polynomial $i + (-1-i) x + x^2$ has no roots within the ball $\{ z \mid \left | z  \right |  < 1 \}$.

\subsection{Remarks} \label{secsub:complex_remark}

Generally, most complex-root-counting procedures boil down to applying the argument principle and approximating some winding number. The Cauchy index on the complex plane elegantly approximates the winding number, and it can be effectively computed by remainder sequences and sign variations as in the application of Sturm's theorem. As this approach is moderately efficient, in 1978 Wilf \cite{Wilf:1978fy} used this counting mechanism for his (complex) root isolation algorithm.

However, as we mentioned earlier in \S\ref{secsub:budan_fourier_remark}, remainder sequences are generally considered too slow for modern computer algebra systems. As a result, in 1992 Collins and Krandick \cite{Collins:1992kf} proposed an approach to directly approximate the winding number through efficient real root isolation based on Descartes' rule of signs, and this approach is actually the one that has been widely implemented in modern systems like Mathematica and SymPy.

In the future, we hope to implement Collins and Krandick's approach. Luckily, the Descartes roots test we mechanised in \S\ref{secsub:descartes_roots_test} should serve as the first step.

\section{Related Work} \label{sec:related_work}

Counting \emph{distinct} real roots with Sturm's theorem has been widely implemented among major proof assistants including PVS \cite{Narkawicz:2015do}, Coq \cite{Mahboubi:2012gg}, HOL Light \cite{mclaughlin-harrison} and Isabelle \cite{Eberl:2015kb,Sturm_Tarski-AFP,Li:2017cg}. In contrast, our previous complex-root-counting procedure \cite{li_cauchy_index} seems to be the only one that counts complex roots, since counting complex roots usually requires a formal proof of the argument principle in complex analysis, which (to the best of our knowledge) is only available in Isabelle/HOL \cite{Li_ITP2016}.

According to the website of \emph{Formalizing 100 Theorems}\footnote{\url{http://www.cs.ru.nl/~freek/100/index.html}}, Descartes' rule of signs has been independently formalised in Isabelle/HOL \cite{Descartes_Sign_Rule-AFP}, HOL Light, and ProofPower, and all three versions seem to follow an informal inductive proof by Arthan \cite{arthan2007descartes}. In Coq, Bertot et al. \cite{berstein_coq} have investigated real root isolation through Bernstein coefficients, and during the investigation they have proved a corollary of Descartes' rule of sign: the polynomial has exactly one positive root if there is only one sign change in its coefficient sequence. In comparison, we have formalised a more general result (i.e., the Budan-Fourier theorem), and derive Descartes' rule of signs as an almost trivial consequence. As a benefit of this more general result, we can additionally derive the corollary that the roots approximation through both Descartes' rule of signs and the Descartes roots test will be exact when all roots are real;  it is not clear how to deduce this without the Budan-Fourier theorem.

Since Thiemann and Yamada have formalised Yun's algorithm in Isabelle/HOL \cite{Polynomial_Factorization-AFP,Thiemann:2016hu}, there could be an alternative procedure to count multiple real roots exactly. Given $P \in \mathbb{R}[x]$, with Yun's algorithm we can have a square-free factorisation of $P$:
\[
	P = Q_1 Q^2_2 Q^3_3 \cdots Q^n_n,
\]
such that polynomials from $\{Q_i\}$ ($1 \leq i \leq n$) are pairwise coprime and square-free. We can then obtain a procedure to count multiple roots by applying Sturm's theorem to each $Q_i$, multiplying the result by $i$, and summing them together. We believe our extended Sturm's theorem will be more efficient than the sketch above, but that is never for sure until we perform a side-by-side comparison. 

Potential applications of our work include various formalisations of algebraic numbers in Coq  \cite{Cohen:2012ba} and Isabelle/HOL \cite{Li_CPP2016,Thiemann:2016wc}. A real algebraic number is usually encoded as a polynomial $P$ and an isolation interval, and this interval is frequently tested (or refined) to guarantee that exactly one root of $P$ lies within it. At present, the testing and refining process relies on Sturm's theorem, which can be replaced by our new Descartes roots test for better efficiency. Furthermore, when encoding complex algebraic numbers, we may need to deal with an isolation box or ball in the complex plane, where our complex-root-counting procedures should be of help.

\section{Experiments} \label{sec:experiments}

In this section, we briefly benchmark our root-counting procedures over some randomly generated polynomials, to convey an idea about their scalability. All the experiments are run on a Intel Core i7 CPU (quad core @ 2.66 GHz) and 16 gigabytes RAM. When benchmarking verified operations, the expression to evaluate is first defined in Isabelle/HOL, and then extracted and evaluated in Poly/ML\@. The reason for this is that when invoking \isa{\isacommand{value}} in Isabelle/HOL to evaluate an expression,  a significant and unpredictable amount of time is spent generating code, so we evaluate an extracted expression to obtain more precise results. 

First, we compare using the classical Sturm theorem, the extended Sturm theorem, the Budan-Fourier theorem, and the Descartes roots test to count/approximate the number of real roots of various polynomials over the interval $(0,1)$ or $(0,1]$. As illustrated in Table \ref{tab:real_roots} and Figure \ref{fig:selected_polynomials}, procedures based on remainder sequences (i.e., Sturm and Ex\_Sturm) are much slower than the others, and their performance degrades rapidly as the bit size of the coefficients grows. In the meantime, the difference in performance between the Budan-Fourier theorem and the Descartes roots test is, surprisingly, marginal. We believe this is due to our naive implementation of Taylor shift, which usually contributes most to the running time of the roots test.

In addition, we also apply our complex-root-counting procedures to count roots within the upper half-plane and the ball $\{ z \mid |z| < 1 \}$. The result is illustrated in Table \ref{tab:complex_roots}: both methods have shown moderate performance, but due to their method of computing remainder sequences the performance deteriorates quickly as the coefficient bit-size increases.

\begin{figure*} [ht]
	\begin{align*}
	\mathrm{P_1(x)=}\quad&-\frac{85}{68}+\frac{70}{5} x+\frac{88}{79} x^{2}+\frac{29}{75} x^{3}+\frac{80}{51} x^{4}-\frac{66}{52} x^{5}+\frac{9}{71} x^{6}-\frac{14}{61} x^{7}-\frac{27}{64} x^{8}-\frac{100}{83} x^{9}+\frac{1}{53} x^{10}-\frac{23}{85} x^{11}+\frac{83}{98} x^{12}+\frac{48}{16} x^{13} \\ & -\frac{89}{25} x^{14}-\frac{100}{5} x^{15}+\frac{36}{28} x^{16}+\frac{1}{1} x^{17}+\frac{43}{99} x^{18}-\frac{29}{32} x^{19}+\frac{74}{97} x^{20}+\frac{9}{5} x^{21}+\frac{20}{70} x^{22}-\frac{89}{27} x^{23}-\frac{33}{48} x^{24}+\frac{16}{33} x^{25}+\frac{84}{63} x^{26} \\ & +\frac{96}{89} x^{27}+\frac{22}{69} x^{28}+\frac{95}{97} x^{29}\\
	\\%
	\mathrm{P_2(x)=}\quad & -34-28 x+5 x^{2}-39 x^{3}+83 x^{4}-89 x^{5}-49 x^{6}+94 x^{7}-66 x^{8}+18 x^{9}+75 x^{10}+84 x^{11}-98 x^{12}-68 x^{13}+12 x^{14}+46 x^{15} \\ & -43 x^{16}+98 x^{17}+24 x^{18}-30 x^{19}+10 x^{20}-88 x^{21}+54 x^{22}+79 x^{23}-29 x^{24}+12 x^{25}-55 x^{26}-46 x^{27}-18 x^{28}+50 x^{29}\\
	%
	\\%
	\mathrm{P_3(x)=}\quad& \frac{9}{30}-\frac{65}{82} x-\frac{94}{68} x^{2}+\frac{9}{33} x^{3}-\frac{56}{83} x^{4}-\frac{22}{35} x^{5}+\frac{73}{31} x^{6}+\frac{69}{2} x^{7}-\frac{58}{43} x^{8}+\frac{71}{22} x^{9}-\frac{75}{44} x^{10}+\frac{2}{49} x^{11}+\frac{24}{40} x^{12}+\frac{33}{62} x^{13} -\frac{17}{2} x^{14} \\ & -\frac{39}{82} x^{15}-\frac{55}{43} x^{16}-\frac{26}{47} x^{17}+\frac{46}{4} x^{18}-\frac{48}{26} x^{19}+\frac{35}{83} x^{20}-\frac{50}{100} x^{21}-\frac{60}{65} x^{22}+\frac{66}{36} x^{23}-\frac{43}{76} x^{24}+\frac{30}{24} x^{25}+\frac{18}{28} x^{26}-\frac{96}{51} x^{27} \\&+\frac{49}{42} x^{28}-\frac{41}{89} x^{29}+\frac{81}{90} x^{30}-\frac{65}{57} x^{31}-\frac{70}{64} x^{32}-\frac{50}{26} x^{33}+\frac{91}{40} x^{34}+\frac{52}{68} x^{35}-\frac{91}{99} x^{36}-\frac{79}{59} x^{37}+\frac{15}{93} x^{38}-\frac{56}{42} x^{39}-\frac{20}{59} x^{40}\\&+\frac{50}{62} x^{41}-\frac{27}{77} x^{42}+\frac{28}{53} x^{43}-\frac{36}{75} x^{44}\\
	\\ %
	\mathrm{P_4(x)=}\quad& -20 -6 x-50 x^{2}-95 x^{3}+35 x^{4}-64 x^{5}+77 x^{6}-56 x^{7}+18 x^{8}-94 x^{9}-74 x^{10}-69 x^{11}-62 x^{12}-93 x^{13}-4 x^{14}-41 x^{15} \\& -47 x^{16}-48 x^{17}-95 x^{18}-41 x^{19}+29 x^{20}+76 x^{21}+70 x^{22}-67 x^{23}-91 x^{24}-93 x^{25}-55 x^{26}-34 x^{27}-67 x^{28}-61 x^{29}\\&-8 x^{30}+32 x^{31}+8 x^{32}-33 x^{33}-27 x^{34}-8 x^{35}+88 x^{36}+53 x^{37}-28 x^{38}-66 x^{39}-72 x^{40}-46 x^{41}+15 x^{42}-19 x^{43}+29 x^{44}	\\
	\\
	\mathrm{P_5(x)=}\quad& (- \frac{93}{47} - \frac{49}{8} i ) +
	(\frac{187}{47} + \frac{547}{88} i) x +
	(- \frac{203}{67}) + \frac{538}{11} i) x^2 +
	( \frac{2}{67} - \frac{1181}{24}) x^3 +
	( \frac{133}{81} - \frac{25}{24}) x^4 +
	 \\ &(\frac{2111}{5670} + \frac{71}{12}) x^5 +
	(\frac{5949}{70} - \frac{64}{15}) x^6 +
	(- \frac{305}{3} + \frac{18}{5}) x^7 +
	(\frac{4067}{255} - \frac{411}{89}) x^8 +
	(- \frac{9}{68} + \frac{2669}{4895}) x^9 +
	(- \frac{3}{20} + \frac{4}{55}) x^{10} \\
	\\
	\mathrm{P_6(x)=}\quad& (51 - 83 i) +
	(- 82 + 29 i ) x +
	(- 37 - 6 i) x^2
	(1 + 45 i ) x^3 +
	(145 - 57 i) x^4 +
	(- 10 + 17 i) x^5 +
	(- 39 + 22 i) x^6 \\ & +
	(40 - 35 i) x^7 +
	(- 112 - 27 i) x^8 +
	(106 - 2 i) x^9 +
	(- 63 +97 i ) x^{10} \\
	\end{align*}
	\caption{Some example polynomials.} \label{fig:selected_polynomials}
\end{figure*}

\begin{table}
	\caption{Applying various procedures to count the number of real roots over an interval.}
	\label{tab:real_roots}
	\begin{minipage}{\columnwidth}
		\begin{center}
			\begin{tabular}{@{}ccccc@{}} \toprule
				& \multicolumn{4}{c}{Time (s)} \\ \cmidrule(r){2-5}
				Polynomial & Sturm  & Ex\_Sturm & Fourier & Descartes\\ \midrule
				$P_1$ & 12.123  & 23.418 & .002 & .002  \\
				$P_2$ & 1.612 & 1.742  & .001 & 0  \\
				$P_3$ & 322.569  & 524.975  & .007 & .007  \\
				$P_4$ & 8.894 & 13.425 & .003 & .003  \\
				\bottomrule
			\end{tabular}
		\end{center}
	\end{minipage}
\end{table}

\begin{table}
	\caption{Counting the number of complex roots within the upper half-plane and a ball.}
	\label{tab:complex_roots}
	\begin{minipage}{\columnwidth}
		\begin{center}
			\begin{tabular}{@{}ccc@{}} \toprule
				& \multicolumn{2}{c}{Time (s)} \\ \cmidrule(r){2-3}
				Polynomial & proots\_upper  & proots\_ball \\ \midrule
				$P_5$ & 4.359  & 24.509  \\
				$P_6$ & 0.633  & 0.256   \\
				\bottomrule
			\end{tabular}
		\end{center}
	\end{minipage}
\end{table}

\section{Conclusion} \label{sec:conclusion}

In this paper, we have strengthened the existing root-counting tools in Isabelle. In particular, we have 
\begin{itemize}
	\item formalised a proof of the Budan-Fourier theorem, and thereby implemented the Descartes roots test,
	\item extended our previous formalisation of the classical Sturm theorem to count real roots with multiplicity,
	\item applied part of the results above to improve our previous complex-root-counting procedures by allowing roots on the border in the half-plane case and providing a procedure to count roots within a ball.
\end{itemize}
The proofs described in this paper are about 6000 LOC in total, and took around 6 person-months to complete.

As counting polynomial roots is a fundamental topic in computer algebra and numerical computing, we believe our verified routines will be of use when certifying continuous systems and for coding tactics.

\begin{acks}                          
	Both Li and Paulson were supported by the ERC Advanced Grant ALEXANDRIA (Project 742178), funded by the European Research Council. We thank Manuel Eberl and Ren\'e Thiemann for various fruitful discussions and for pointing us to the approach of counting multiple roots through square-free factorisation. We are also grateful for suggestions from anonymous reviewers.	 
\end{acks}

\bibliography{bibfile}



\end{document}